\documentclass[a4paper]{article}

\usepackage{float}
\usepackage{fullpage}
\usepackage{graphicx}
 \usepackage{tkz-graph,cancel}
\usepackage{xcolor}
\usepackage{amsmath}
\usepackage{amssymb}
\usepackage{amsthm}
\usepackage[sort,nocompress]{cite}
\usepackage[ruled,vlined,lined,boxed,commentsnumbered]{algorithm2e}
\usepackage{multirow}
\usepackage{todonotes}
\usepackage{subfigure}
\SetKwComment{Comment}{$\triangleright$\ }{}

\theoremstyle{plain}
\newtheorem{theorem}{Theorem}[section]
\newtheorem{corollary}[theorem]{Corollary}
\newtheorem{lemma}[theorem]{Lemma}

\newtheorem{claim}{Claim}

\theoremstyle{definition}
\newtheorem{definition}{Definition}
\theoremstyle{remark}
\newtheorem{rem}[theorem]{Remark}
\newtheorem{example}{Example}
\newtheorem*{notation}{Notation}

\usepackage{pgf}
\usepackage{tikz}
\usetikzlibrary{arrows,automata}

\usepackage{authblk}

\newcommand{\fc}[1]{\text{Fact}(\mathcal #1)}
\newcommand{\pf}[1]{\text{Pref}(\mathcal #1)}
\newcommand{\pff}[1]{\text{\emph{Pref}}(\mathcal #1)}
\newcommand{\sff}[1]{\text{\em Suff}(\mathcal #1)}

\title{Wheeler Languages}

\author[1]{Jarno Alanko\thanks{jarno.alanko@helsinki.fi}}
\author[2]{Giovanna D'Agostino\thanks{giovanna.dagostino@uniud.it}} 
\author[2]{Alberto Policriti\thanks{alberto.policriti@uniud.it}}
\author[3]{Nicola Prezza\thanks{nprezza@luiss.it}} 

\affil[1]{University of Helsinki, Finland} 
\affil[2]{University of Udine, Italy} 
\affil[3]{Luiss Guido Carli, Rome, Italy}

\begin{document}
 
\maketitle

\tableofcontents

\newpage 

\section*{Introduction}

The Burrows-Wheeler Transform (BWT) of a given string is an invertible transformation with many important and deep properties (see \cite{Burrows94ablock-sorting}). It can be computed on a given string by marking the beginning of the string by the special character \# and reading the first column of the matrix consisting of the co-lexicographically ordered circular permutations of the string (BW-matrix, see Figure \ref{fig-ab}-(a))\footnote{In the ``official'' definition of the transform, the \emph{lexicographic} ordering of circular permutations and a \$-mark of the \emph{end} of the string are used. Working with the co-lexicographic ordering is a bit more natural while studying formal languages and does not make any significant difference.}. 

The fact that the transform is invertible can be seen as one of its most basic and useful features, and it is a consequence of the fact that the BWT (actually the BW-matrix) enjoys the so-called First-Last property (FL-property, more on this below). Being invertible and, at the same time, rich of single-letter runs induced by the co-lexicographic order of prefixes\footnote{The co-lexicographic order of prefixes can be read on the right side of the BW-matrix.}, the BWT becomes the basis for a family of tools needing very little extra data-structures (see \cite{DBLP:journals/csur/NavarroM07}).

The FL-property consists in the observation that in the first and last columns of the BW-matrix, the relative order of different occurrences of the same  character is maintained. Consider, for example, the BWT of the string {\em \#banana}, that is {\em bnn\#aaa}, and notice that the First-Last property can be used to instruct us on how to reconstruct  {\em \#banana}:  
start from \# on the first column, search the occurrence of \# in the last column, move to the first column on the same row, and continue with the corresponding character (i.e. $b$, see Figure \ref{fig-ab}-(b)). 
The correctness of the reconstruction of the original string is a consequence of the FL-property: at each step the character read on the first column corresponds to the one determined on the last column.

\begin{figure}[h!]%
    \centering
    \subfigure[The BWT of the string \emph{\#banana} can be read in the first (F) column of the BW-matrix. Character \# is the lexicographic smallest.] {
      $
\begin{array}{ccccccc}
 F&&&&&&L\\
 b&a&n&a&n&a&\# \\
 n&a&n&a&\#& b&a \\
 n&a&\#& b&a&n&a \\
 \#& b&a&n&a&n&a \\
 a&n&a&n&a&\#& b \\
 a&n&a&\#& b&a&n \\
a&\#& b&a&n&a&n \\
  \\
\end{array} 
    $}%
    \qquad
    \subfigure[The path (automata) encoding the actions to be performed on the BWT (column F) in order to reconstruct the original string, starting from \#.]{
\begin{tikzpicture}[->,>=stealth', semithick, initial text={}, auto, scale=.7]
 \node[state, label=above:{}] (B) at (0,0) {$b$};
 \node[state, label=above:{}] (N1) at (2,0) {$n$};
 \node[state, label=above:{}] (N2) at (4,0) {$n$};
 \node[state,initial , label=above:{}] (S) at (6.5,0) {$\#$};
 \node[state, label=above:{}] (A1) at (8.5,0) {$a$};
 \node[state, label=above:{}] (A2) at (10.5,0) {$a$};
 \node[state, label=above:{}] (A3) at (12.5,0) {$a$};

\draw (S) edge [bend left=30, above] node [bend right, above] {} (B);
\draw (B) edge [bend left=30, below] node [bend right, below] {} (A1);
\draw (A1) edge [bend left=30, above] node [bend right, above] {} (N1);
\draw (N1) edge [bend left=30, above] node [bend right, below] {} (A2);
\draw (A2) edge [bend left=30, above] node [bend right, above] {} (N2);
\draw (N2) edge [bend left=30, below] node [bend right, below] {} (A3);
\end{tikzpicture}
}%
    \caption{Starting from \# in the first (F) column of the BW-matrix in (a), reading the character, and moving to the corresponding position in the last (L) column, the original string can be reconstructed. The full procedure is encoded in the path (linear automaton) in (b).}%
    \label{fig-ab}%
\end{figure}
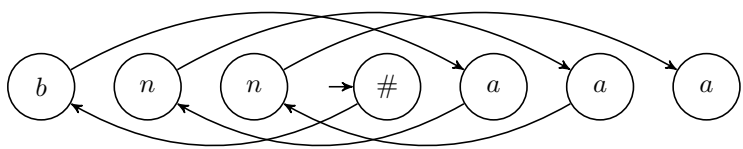

The above graph can be seen as a very simple (linear) state-labelled finite automaton, with node labels organized in the order they appear in the F column (the BWT). With a slight twist, let us now use a different ordering: the one induced by the L column of the BW-matrix. 
The result, reflecting on the linear automaton the nice computational features of the BWT, is depicted in Figure \ref{fig:wheeler path}.  

\begin{figure}[h!]
 \begin{center}
\begin{tikzpicture}[->,>=stealth', semithick, initial text={}, auto, scale=.8]

 \node[state,initial, label=above:{}] (S) at (0,0) {$\#$};
 \node[state, label=above:{}] (B) at (8,0) {$b$};
 \node[state, label=above:{}] (A1) at (2,0) {$a$};
 \node[state, label=above:{}] (N1) at (10,0) {$n$};
 \node[state, label=above:{}] (A2) at (4,0) {$a$};
 \node[state, label=above:{}] (N2) at (12,0) {$n$};
 \node[state, label=above:{}] (A3) at (6,0) {$a$};

\draw (S) edge [bend left=30, above] node [bend right, above] {} (B);
\draw (B) edge [bend left=30, below] node [bend right, below] {} (A1);
\draw (A1) edge [bend left=30, above] node [bend right, above] {} (N1);
\draw (N1) edge [bend left=30, above] node [bend right, below] {} (A2);
\draw (A2) edge [bend left=30, above] node [bend right, above] {} (N2);
\draw (N2) edge [bend left=30, below] node [bend right, below] {} (A3);
\end{tikzpicture}\caption{The path automaton of Figure \ref{fig-ab}, reorganized according to the order of column L.}\label{fig:wheeler path}
\end{center}
\end{figure}
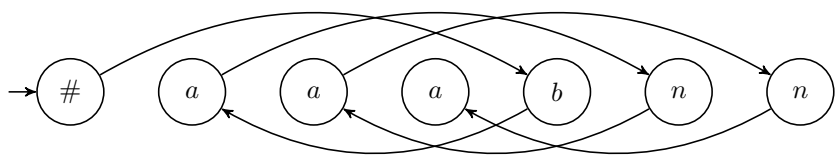
 
In a sense, this layout seems more natural as it orders nodes according to the co-lexicographic ordering of the strings read from the source to each of the nodes. 
%
%
%
The graphs we obtain in this way are precisely paths
encoding the procedure inverting the Burrows-Wheeler transform
of a given string. Much more interestingly,  one may ask the following question: can we generalise our considerations to  the  context of general \emph{ordered} graphs (i.e. not being necessarily  paths)? One may wonder which properties of graphs/automata/orderings enforce the above behaviour.

The objects resulting from this analysis are Wheeler graphs \cite{gagie2017wheeler} and their characterising  properties---working for general ordered graphs---are: 
\begin{itemize}
    \item[ (i)] the ordering of character-labelled states must be coherent with  an  (a priori fixed) order    of characters, and
    \item[(ii)] the ordering of states $ u,v $  bearing the same character-label must be coherent with the ordering of all the predecessor pairs $ u',v' $  with associated arcs $ (u',u), (v',v) $.
\end{itemize}  

\medskip

The main application of Wheeler graphs is that they admit an efficient index data structure for searching for subpaths with a given path label \cite{gagie2017wheeler}. This is in contrast with recent results showing that in general, the subpath search problem can not be solved in subquadratic time, unless the strong exponential time hypothesis is false \cite{equi2019complexity}. The indexing version of the problem was also recently shown to be hard, unless the orthogonal vectors hypothesis is false \cite{equi2020conditional}. The strong exponential time hypothesis implies the orthogonal vectors hypothesis.

In the big picture, Wheeler graphs lift the applicability of the Burrows-Wheeler transform from strings to languages. 

In this paper we study the regular languages 
accepted by automata having a Wheeler graph as transition function.
The study is carried out in both the deterministic and the non-deterministic case and shows that
 Wheeler Automata establish a deep link between  intervals of states---in the Wheeler ordering imposed by the definition---and ``intervals''  
of strings---in the co-lexicographic ordering of prefixes of elements in $ \mathcal L $.  
Our investigation starts from  some results already appeared in \cite{alanko2020regular}, where we proved 
 that the classic characterisation of regular languages based on Myhill-Nerode Theorem can be generalised and adapted to the Wheeler case. The generalisation is proved by introducing equivalence
 classes which are convex sets in the co-lexicographic ordering of prefixes of strings in $ \mathcal L $.  This characterization allows also to prove 
 that the (potential) exponential blow-up in the number of states observed in general  when passing from a 
non-deterministic to a deterministic automaton, cannot take place in the Wheeler case. 
In this paper we apply  these results (which we add with complete proofs for the sake of completeness and readability)  
to find a solution to  the problem  of  effectively testing  for Wheelerness  languages given by a deterministic or non-deterministic automaton.
 In addition, in the deterministic case we can show that the test takes polynomial time. 
The results on testing Wheelerness are based on a theorem that characterises minimal deterministic automata accepting Wheeler languages on a purely graph-theoretic property.

Next, we take the automata's point of view on Wheelerness. More specifically, since   the problem of deciding whether a given NFA can be endowed with a Wheeler order  is obviously decidable, we tackle its  complexity which, although polynomial in special cases (see \cite{alanko2020regular}), is known to be NP-complete in the general case (see \cite{DBLP:conf/esa/GibneyT19}). Here we prove that  over a natural   subclass of NFA, the \emph{reduced} ones---that is, those in which no two states are reachable by the same set of strings---,  the problem can in fact be solved in polynomial time.

Finally, we take a closer look at classical operations among  Wheeler Languages. 
Since Wheeler languages are  a subclass of the class of Ordered Languages (see \cite{thierrin1974ordered_aut}),  they  are star-free, namely they can be generated from finite languages by boolean operations and compositions only. As such,  they are  definable in   the first order theory of linear orders $FO(<)$ (see \cite{DBLP:conf/birthday/DiekertG08} for a survey on FO-definable languages). However, as we shall see, there are very few classical operations  preserving Wheelerness.  
While  regular languages are closed for boolean and regular operations, we  prove that,   with a few exceptions,  this is not  the case  for Wheeler languages. 

\medskip

The paper is organised as follows. 
Section \emph{Basics} contains basic notions and notations.  Section \emph{Wheeler Automata and Covex Sets} introduces the notion of Wheeler Automata and links the natural linear orderings definable on states and strings, respectively. In this section we also  introduce convex equivalences,    allowing us to prove a precise ``Wheeler version'' of the classical Myhill-Nerode Theorem for regular languages.  
Section \emph{Testing Wheelerness} tackles the problem, discussed in two separate subsections, of whether a given language  or a given automaton  is Wheeler. 
The next section,  \emph{Closure Properties for Wheeler Languages},    considers regular operations and   closure properties that are known to  hold for regular languages  and  checks whether they also hold for Wheeler languages. In this section we further consider intervals on the co-lexicographic order, proving that they are Wheeler.  
We conclude the paper with the section  \emph{Conclusions and Open Problems}.

 \section{Basics}
\subsection{Automata}
 If $\Sigma$ is a finite alphabet, we denote by $ \Sigma^*$ ($ \Sigma^+$)  the set of  (non-empty) finite words over $\Sigma$.  If  $\mathcal L \subseteq \Sigma^*$ we denote by $ \text{\em Pref}(\mathcal L),  \text{\em Suff}(\mathcal L), $ and $ \text{\em Fact}(\mathcal L)$ the set of prefixes, suffixes, and factors  of strings in  $\mathcal L$, respectively. More formally:
 \begin{align*}
 \text{\em Pref}(\mathcal L)  & = \{\alpha: \exists \beta \in \Sigma^*~ \alpha \beta\in \mathcal L\}, &  \text{\em Suff}(\mathcal L)  & = \{\beta: \exists \alpha \in \Sigma^*~ \alpha \beta\in \mathcal L\}, & \text{\em  Fact}(\mathcal L)  & = \{\alpha: \exists \beta, \gamma \in \Sigma^*~ \gamma \alpha \beta\in \mathcal L\}.
 \end{align*}

In the following we will denote by  $\mathcal A=(Q, s, \delta,F)$ a finite automaton  (an NFA) accepting strings in $ \Sigma^{*} $, with $ Q $ as set of states, $ s $ \emph{unique} initial state with no incoming transitions, $ \delta(\cdot, \cdot): Q \times \Sigma \rightarrow \mathcal Pow(Q) $ transition function, and $ F\subseteq Q $ final states. Note that assuming that $s$ has no incoming transitions is not restrictive, as any NFA can be made to satisfy this condition by just duplicating $s$ into an initial $s'$ with no incoming transitions and a non-initial $s''$ with all the incoming transitions of the original $s$.

An automaton  $ \mathcal A $  is {\em deterministic}   (a DFA),  if $ |\delta(q,a)|\leq 1 $,  for any $ q \in Q $ and $ a \in \Sigma $. As customary, we extend $ \delta $ to operate on strings as follows: for all $ q \in Q, a\in \Sigma, $ and $ \alpha \in \Sigma^{*} $:
\begin{align*}
 {\delta}(q,\epsilon)=\{q\}, &   \hspace{1cm}
 {\delta}(q,\alpha a) = \bigcup_{v\in{\delta}(q,\alpha)} \delta(v,a).
\end{align*}   
If the automaton is deterministic we  write  ${\delta}(q,\alpha)=q'$  for   the unique $q'$  such that ${\delta}(q,\alpha)=\{q'\}$  (if defined). 
  We denote by   $\mathcal L(\mathcal A)=\{\alpha\in \Sigma^*: \delta(s, \alpha)\cap F\neq \emptyset\}$   the language accepted by the automaton $ \mathcal A $. $ \mathcal A $ is dubbed \emph{complete} if for any $ q \in Q, a\in \Sigma $, $ \delta(q,a) $ is defined. In general, we do not assume $ \delta $ to be {complete}---to see why, wait for Example \ref{incomplete} below---, while we do assume that each state can reach a final state   and also that every state is reachable from the (unique) initial state. 
Hence, $ \text{\em  Pref}(\mathcal L(\mathcal A))$, the collection of prefixes of words accepted by $ \mathcal A $,  consists of the set of  words that can be \emph{read} by $ \mathcal A $.

Using the terminology from~\cite{alanko2020regular}, an {\em  input-consistent automaton} is such that every state has incoming edges labeled by the same character. This class of automata  is the one considered in the original definition  of Wheeler graph in \cite{gagie2017wheeler}. It is  fully general:  any automaton can be converted into an input-consistent one recognizing the same language at the price of increasing $|Q|$ by a multiplicative factor $|\Sigma|$ \cite{alanko2020regular}. 
Moving labels from an edge to its target state, input-consistent automata  can  be  described  as   \emph{state-labeled} automata (see Example (\ref{exw})). In this paper we will therefore use the term \emph{state-labeled} in place of \emph{input-consistent}. Given a state-labeled automaton, we denote by $\lambda : Q \rightarrow \Sigma \cup \{\#\}$ the function that returns the (unique) label of a state, so that  $\delta(u,c)$ is the set of $c$-labelled successors  of $u$. 
To make $\lambda$ complete and to be consistent with the definition of Burrows-Wheeler Transform, we assign $\lambda(s) = \# \notin \Sigma$, where $\#$ is a character not labeling any other state.
When for all  $u, v \in C\subseteq Q$ we have $\lambda(u) = \lambda(v)$, let $\lambda(C)$ be the unique character $c = \lambda(u)$, for any $u\in C$. 
To make notation consistent between edge-labeled and state-labeled automata, given a path   $v_0, \ldots, v_n$ we define 
its label as $\lambda(v_1)\ldots\lambda(v_n)$, 
so that the first  node $v_0$ does not contribute to the string labeling the path.  
All our results dealing with Wheeler automata will use state-labeled automata. In other results, however, we will need to work with standard edge-labeled automata. In this case, we will explicitly say that the automaton is edge-labeled and use the notation $\lambda(u,v)\in \Sigma$ to denote the label of an automaton's edge $(u,v)$ (note that, in the case of edge-labeled automata, no edge is labeled with $\#$).

 
 \subsection{Convex Sets}
 As we shall see,  Wheeler automata and languages  naturally lead to considering \emph{convex subsets} of a linear order. We collect here a few definitions and results that will turn out handy while reasoning on convex sets. 

 \begin{definition}  
 Consider a linear order $(L, <)$.
 \begin{enumerate}
 \item A \emph{convex set}  in  $(L, <)$  is a  $I\subseteq L$  such that
 \[(\forall x,x'\in I)(\forall y \in L) (x< y< x' \rightarrow y\in I).\]
 \item Given $I,J$   convex  in  $(L, <)$ and  $I\subseteq J$, then: 
 \begin{itemize}
\item[-] $I$ is a {\em prefix} of $J$ if $(\forall x \in I )( \forall y \in J\setminus I) ( x< y)$;
\item[-]  $I$   is a    {\em suffix} of $J$ if   $(\forall y \in J \setminus I )( \forall x \in I )( y< x)$.
\end{itemize}
\item A family $\mathcal C$ of non-empty convex sets in  $(L, <)$ is said to have the  {\em prefix/suffix property} if, for all $I,J \in \mathcal C$ such that  $I \subseteq J$, $I$ is either a prefix or a suffix  of $J$.
 \end{enumerate}
  \end{definition}
  
  In particular, if $a,b\in L$ for a linear order $(L,<)$, then we denote by $[a,b]$ the convex set:
  \[ [a,b]=\{c\in L: a\leq c\leq b\}.\]
  $[a,b]$ is called the {\em closed interval based on $a,b$};   other  kinds of intervals, denoted by $(a,b), (a,b], (-\infty, b), \ldots$ are defined as usual. Notice that any  convex set $I$ having  a maximum   and a minimum     is an interval: 
  \[I=[min_< I, max_< I].\] In particular, all convex subsets of a finite linear order are intervals, and we shall use freely both names for them. 
  
  The most convenient feature of a family $ \mathcal C $ enjoying the prefix/suffix property, is the fact that its elements can be easily ordered. 
  
 \begin{definition}
 Let $\mathcal C$ be a family of non-empty convex sets  of a linear order  $(L, <)$ having the prefix/suffix property. Let $<^{i}$ (or simply $ < $) the binary relation over $\mathcal C$ defined by
 \begin{align*}
 I<^{i} J &\Leftrightarrow ( \exists x \in I)( \forall y\in J )( x< y )\lor (\exists y \in J)( \forall x \in I )( x< y).
 \end{align*}
 \end{definition}
The following lemma is easily proved.
 \begin{lemma} \label{convex+order}  
$ (\mathcal C,<^{i}) $ is a strict linear order.  
 \end{lemma}
 
Note that whenever  any non-empty convex set $I$  has minimum $m_I$  and maximum $M_I$---which is the case, for example, when the linear order   $(L, <)$ is finite---, the above order  $<^{i}$ can be equivalently described on a family having the prefix/suffix property, by: 
\[
I  <^{i} J   ~\Leftrightarrow ~(m_I <  m_J) \lor [ (m_I=m_J) \wedge (M_I < M_J) ]~ \Leftrightarrow~ m_I+M_I<m_J+M_J. 
\]

  The following lemma will allow us to bound (linearly) the blow-up of the number of states taking place when moving from a  Wheeler NFA to a Wheeler DFA (see Definition \ref{WNFA_WDFA} below). 
\begin{lemma}\label{2n}
	Let $(L, <)$ be a finite linear order of cardinality $|L|=n$, and let  $\mathcal C$ be a {\em prefix/suffix} family of non-empty convex sets in $(L, <)$. Then:
	\begin{enumerate}
		\item $|\mathcal C| \leq 2n-1$.
		\item The upper bound is tight: for every $n$, there exists a prefix/suffix family of size $2n-1$.
	\end{enumerate}
\end{lemma}
\begin{proof}
 (1) Since $L$ is finite, for any $I, J\in \mathcal C$ we 
 have 
 \[I< J ~\Leftrightarrow ~ m_I+M_I<m_J+M_J\]
 which implies 
 \[I\neq J ~\Leftrightarrow ~ m_I+M_I\neq m_J+M_J.\]
Since  the possible values of $m_I+M_I$,  for  $I \in \mathcal C$,
range  between $2$ and $2n$, the bound $|\mathcal C| \leq 2n-1$  follows. 

(2) Consider the prefix/suffix family containing just one maximal interval and all its proper prefixes and suffixes:
$\mathcal C = \{  L[1,n], L[1,1], \dots, L[1,n-1], L[2,n], \dots, L[n,n]\}$. This family satisfies $|\mathcal C| = 2n-1$. \hfill
\end{proof} 

 \begin{definition} Consider a linear order $ (L, <) $ and an equivalence relation $ \sim $ over its domain $ L $.
 \begin{enumerate}
 \item We say that $ \sim $ is {\em convex} if its equivalence classes are convex sets in  $(L,<)$.
 \item The {\em convex refinement}  of  $\sim$ over   $(L, <)$,  is  the relation $\sim^{c}$  on $ L $ defined as follows. For all $a,b\in L$:
 \begin{align*}
 a \sim^{c} b & \Leftrightarrow a\sim b \wedge (\forall d \in L) (min\{a, b\}< d <max\{ a, b\} \rightarrow a\sim d).
 \end{align*}
 \end{enumerate}
\end{definition}

 \begin{lemma} \label{wheeler_refinement} 
 The convex refinement $ \sim^{c} $ of an  equivalence relation $\sim$ over  $(L, <)$,  is a convex equivalence relation. 
 \end{lemma}

 In this paper, if $\Sigma$ consist of a finite number of letters ordered by $ \prec $,  we denote, again by $\prec$, the \emph{co-lexicographic} order over $\Sigma^*$, defined  for $\alpha=a_1\ldots a_n, \beta=b_1 \ldots b_k $, as:
 \begin{align*} 
 \alpha\prec \beta & \Leftrightarrow (n < k \wedge (\forall j \leq n) (a_{n-j} = b_{k-j})) \lor (\exists i) (a_{n-i}\prec b_{k-i}  \wedge (\forall j<i ) ~ a_{n-j}= b_{k-j}).
 \end{align*}

\section{Wheeler Automata and Convex Sets} 

Wheeler languages will be defined below to be regular languages accepted by Wheeler automata, that is,  automata equipped with an ordering among states. It will be proved in \ref{MN-subsection} that Wheeler languages are naturally given as finite families of non-empty convex sets on  $ \prec $  enjoying the prefix/suffix property. 

\medskip

Let us begin giving the definition of Wheeler automaton.

 \begin{definition}\label{WNFA_WDFA}
 A Wheeler NFA (WNFA) $\mathcal A=(Q, s, \delta,<,F)$ is an NFA endowed with a binary relation $ < $, such that:
$ (Q,<) $ is a linear order  having the initial state $s$  as minimum, $ s $ has no  in-going edges, and the following  two (Wheeler) properties are satisfied. Let $ v_{1} \in \delta(u_{1},a_{1}) $ and $ v_{2} \in \delta(u_{2},a_{2}) $:
 \begin{enumerate}
 
	\item[(i)] $ a_{1}\prec a_{2}\rightarrow v_{1} < v_{2} $;
	\item[(ii)] $ (a_{1}=a_{2}  \wedge u_{1} < u_{2}) \rightarrow v_{1} \leq v_{2}$.
\end{enumerate} 
 A Wheeler DFA (WDFA) is a WNFA in which the cardinality of $ \delta(u,a) $ is always less than or equal to one.
\end{definition}

\begin{rem}
A consequence of Wheeler property (i) is that $ \mathcal A $ is \emph{input-consistent}, that is all transitions entering a given state $ u\in Q $ bear the same label: if $ u \in \delta(v,a)$  and $u \in \delta(w,b) $, then $ a=b$.
\end{rem}

On the grounds of the above remark, when drawing Wheeler automata we ``move'' labels from edges to nodes and therefore deal with state-labeled automata: all edges entering  a node labelled $e\in \Sigma$ would then be $e$-edges. 
As mentioned in the introduction, to make $\lambda$ complete we set $\lambda(s) = \# \notin \Sigma$, where $\#$ labels just $s$.

Unless explicitly stated, if we use an alphabet  $\Sigma$ containing alphabetical letters,  we implicitly suppose $\Sigma$   ordered  alphabetically. 
\begin{example}  \label{exw}
The following automaton proves that the  language $ax^*b|zx^*d$ is Wheeler (states ordered from left to right): 

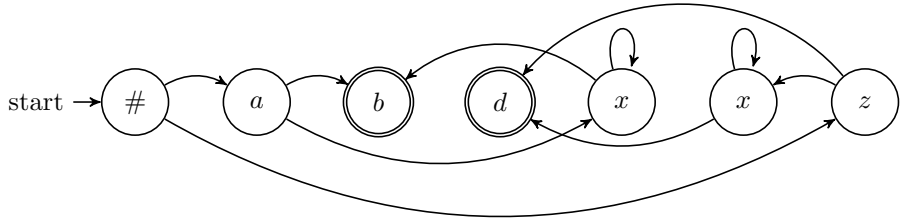
\begin{figure}[h!] 
 \begin{center}
\begin{tikzpicture}[->,>=stealth', semithick, auto, scale=.8]

 \node[state,initial, label=above:{}] (S)    at (-4,0)	{$\#$};
 \node[state, label=above:{}] (A)    at (-2,0)		    	{$a$};
 \node[state, accepting, label=above:{}] (B)    at (0,0)	{$b$};
 \node[state, accepting, label=above:{}] (D)    at (2,0)	{$d$};
 \node[state, label=above:{}] (X1)    at (4,0)		{$x$};
 \node[state, label=above:{}] (X2)    at (6,0)		{$x$};
 \node[state, label=above:{}] (Z)    at (8,0)		    	{$z$};
 
\draw (S) edge [bend left, above] node [bend right, above] {} (A);
\draw (A) edge [bend left, above] node [bend right, above] {} (B);
\draw (A) edge  [bend right=30, below] node [bend left=30, below] {} (X1);
\draw (S) edge  [bend right=30, below] node [bend left=30, below] {} (Z);
\draw (X1) edge  [bend right=40, below] node [bend left=40, below] {} (B);
\draw (X1) edge  [loop above] node {} (X1);
\draw (X2) edge  [loop above] node {} (X2);
\draw (Z) edge  [bend right, above] node [above] {} (X2);
\draw (Z) edge  [bend right=50, above] node [above] {} (D);
\draw (X2) edge  [bend left, below] node [below] {} (D);
\end{tikzpicture}

 \end{center}
\end{figure}
\end{example}

A key consequence of (i) and (ii) above (already proved in \cite{gagie2017wheeler}), is the fact that the set of states reachable in a WNFA $ \mathcal A $ while reading a given string $ \alpha $  is an interval in $(Q,<)$. This important fact will be re-proved below---in Lemma \ref{convex_sets}---, together with what we may call a  sort of its   ``dual'', that is,  the the set of strings read while reaching a given state is a convex set. More precisely, if  $\mathcal A=(Q, s, \delta,<,F)$ is a WNFA, $ u\in Q $, and $\alpha\in \Sigma^*$,     let 
$I_\alpha=\delta(s, \alpha),  I_u=\{\alpha: \delta(s,\alpha)=u\};$ then it easily follows that 
\begin{center}
$ \alpha \in I_u  \text{ if and only if } u \in I_\alpha$,
\end{center}
and in Lemma \ref{convex_sets} we shall prove that $ I_\alpha$ is a convex set in    $ (Q,<) $, while  $ I_u$ is  convex  in $(\text{\em  Pref}(\mathcal L(\mathcal A)),\prec)$.

\medskip

Preliminary to our result is the following lemma.
 
 \begin{lemma}\label{prec_versus_minus} \cite{alanko2020regular}
 If $\mathcal A=(Q, s, \delta,<,F)$  is a   WNFA,   $u,v\in Q$   are states, and $\alpha, \beta \in \text{Pref}(\mathcal L(\mathcal A))$, then:
 \begin{enumerate}
\item  if  $\alpha\in I_u, \beta\in I_v$, and  $\{\alpha,\beta\}\not\subseteq I_v\cap I_u$, 
 then $\alpha\prec \beta  \text{ if and only if } u<v$;
\item 
 if  $u\in I_\alpha, v\in I_\beta $, and  $\{u,v\}\not\subseteq I_\beta \cap I_\alpha$,
 then $\alpha\prec \beta \text{ if and only if } u<v$.
 \end{enumerate}
 \end{lemma}
 
 \begin{proof}  \
 
\begin{enumerate}
\item[(1)]  Suppose   $\alpha\in I_u, \beta\in I_v$ and  $\{\alpha,\beta\}\not\subseteq I_v\cap I_u$. From this we have that  $\alpha\in I_u\setminus I_v$  or   $\beta\in I_v\setminus I_u$,  hence  $u\neq v$ and $\alpha\neq \beta$ follows.

If $u=s$ or $v=s$, either $ \alpha $ or $ \beta $ is the empty string $ \epsilon $ and   the result follows easily. Hence, we suppose $u\neq s \neq v$ and (hence)  $\alpha\neq \epsilon \neq \beta$. 
 
To see the left-to-right implication, assume $\alpha \prec \beta$: we prove that $u<v$ 
 by induction on the maximum betwewn $|\alpha|$ and $|\beta|$. If $|\alpha| =|\beta|=1$, then the property follows from the Wheeler-(i). If $\max(|\alpha|, |\beta|)>1$ and $ \alpha $ and $ \beta $  end with different letters, then again the property follows from  Wheeler-(i). Hence, we are just left with the case in which $\alpha=\alpha' e$ and $\beta=\beta' e$, with $e\in \Sigma$.
Since $\alpha\prec \beta$,   we have $\alpha'\prec\beta'$. Consider  states
$u',v'$ such that $\alpha' \in I_{u'}, \beta'\in I_{v'}$, and $u\in \delta(u',e), v\in \delta(v',e)$. 
Then $\alpha'\in I_{u'}\setminus I_{v'}$ or  $\beta'\in I_{v'}\setminus  I_{u'}$ because otherwise we would have  $\alpha'\in I_{v'}$   and $\beta'\in I_{u'}$   which imply respectively $\alpha\in I_v$ and $\beta\in I_u$.  By induction we have $u'<v'$ and therefore, by Wheeler-(ii), $u\leq v$. From $u\neq v$ it follows  $u<v$.   

 Conversely, for the right-to-left implication, suppose   $u<v$.  Since  $\alpha \neq \beta$,  if it were $\beta\prec \alpha$ then, by the above, we would have $v<u$:  a contradiction. Hence, $\alpha\prec \beta $ holds. 
\item[(2)]  Recall that, by definition,  $ \alpha \in I_u  \text{ if and only if } u \in I_\alpha$ and $ \beta \in I_v  \text{ if and only if } v \in I_\beta$. Hence, the hypothesis that $u\in I_\alpha, v\in I_\beta$ and $\{u,v\} \not\subseteq I_\beta \cap I_\alpha$, is equivalent to say that $ \alpha\in I_u, \beta\in I_v$ and $\{\alpha,   \beta\} \not\subseteq I_v \cap I_u$.
  Therefore, (2) follows from (1).
\end{enumerate} 
 \end{proof}
 
 The following corollary, to be be used in Section \ref{Lwheeler?}, observes that the sequence of states reached  in a WDFA while reading a monotone sequence of strings, must ``stabilise'' to some specific state.  As a matter of fact, it will be proved  in Lemma \ref{compare}  that a similar  property holds also for a WNFA.
 \begin{corollary}\label{monotone}\cite{alanko2020regular}
If $\mathcal A=(Q, \delta, q, <,F)$ is a WDFA,  then, for all $\alpha, \beta \in \pf {L(\mathcal A)}$ it holds 
\begin{align*}
\alpha \prec \beta \Rightarrow \delta(s, \alpha)\leq \delta(s, \beta), &\text{ and }  \delta(s, \alpha)< \delta(s, \beta) \Rightarrow \alpha \prec \beta
\end{align*}
Moreover,  any sequence of states $ (\delta(s, \alpha_{i}))_{\geq 1} $ for $ (\alpha_{i})_{\geq 1} $  monotone sequences   in $(\pf {L(\mathcal A)}, \prec)$, is eventually constant. More precisely, if   $(\alpha_i)_{i\geq 1}$ is a sequence  in $(\pf {L(\mathcal A)}, \prec)$   such that either 
\begin{align*}
\alpha_1\preceq  \alpha_2\preceq \ldots\preceq \alpha_i  \preceq \ldots & \text{ or } 
\alpha_1\succeq \alpha_2\succeq\ldots\succeq \alpha_i  \succeq \ldots 
\end{align*}
 then there exists $u\in Q$ and $n\geq 1$ such that $\delta(s, \alpha_h)=\delta(s,\alpha_k)=u$, for all $h,k\geq n$. 
 \end{corollary}
\begin{proof}   If $\mathcal A=(Q, \delta, q, <,F)$ is a WDFA and $\alpha\in  \pf {L(\mathcal A)}$ then, for all $u\in Q$,  it holds
\[\alpha \in I_u \Leftrightarrow u=\delta(q,\alpha),\]
and \[\alpha \prec \beta \Rightarrow \delta(s, \alpha)\leq \delta(s, \beta), ~~ \text{and}~~  \delta(s, \alpha)< \delta(s, \beta) \Rightarrow \alpha \prec \beta\] easily follows from the previous lemma. 
If  $(\alpha_i)_{i\geq 1}$ is a  monotone sequence  in $(\pf {L(\mathcal A)}, \prec)$,  then the first implication above  implies that  $(\delta(s, \alpha_i))_{i\geq 1}$ is a monotone sequence in $(Q,<)$. Since $Q$ is a finite set, the corollary follows. 

\end{proof}

The following lemma refers to WNFA and proves that  the collection of states reached reading a given string, turns out to be  an interval in the Wheeler order of states. WDFA can be seen a particular case in which intervals degenerate in a single state.    Let $ I_{Q} =\{I_u:~ u\in Q\}$ and  $I_ {\text{\em Pref}(\mathcal L(\mathcal A))}  =\{I_\alpha :~\alpha\in \text{\em Pref}(\mathcal L(\mathcal A)) \}$. 

  \begin{lemma}\label{convex_sets}\cite{alanko2020regular}
 If $\mathcal A=(Q, s, \delta, <,F)$  is a  WNFA  and $\mathcal L = \mathcal L(\mathcal A)$, then:
 \begin{enumerate}
 \item  for all  $u\in Q$,  the set  $I_u$  is  convex   in $( \text{Pref}(\mathcal L(\mathcal A)), \prec)$; 
 \item $ I_{Q}$ is  a prefix/suffix family  of  convex sets  in  $( \text{Pref}(\mathcal L(\mathcal A)), \prec)$;
 \item   for all  $\alpha \in \text{Pref}(\mathcal L(\mathcal A))$,  the set  $I_\alpha$    is  an interval  in $(Q,<)$ (already proved in \cite{gagie2017wheeler}); 
 \item $I_ {\text{Pref}(\mathcal L(\mathcal A))}$  is a prefix/suffix family  of intervals in $(Q,<)$.
 \end{enumerate}
 \end{lemma}

 \begin{proof} \ 
 
 \begin{enumerate}
\item  Suppose $\alpha\prec \beta\prec\gamma$ with $\alpha, \gamma\in I_u$ and $\beta \in  \text{\em Pref}(\mathcal L(\mathcal A))$;  we want to prove that $\beta \in I_u$. From $\beta\in  \text{\em Pref}(\mathcal L(\mathcal A))$ it follows that  there exists  a state $v$ such that
 $\beta \in I_v$.  Suppose, for contradiction, that $\beta\not \in I_u$. Then $\beta \in I_v\setminus I_u$ and from $\alpha\prec \beta$ and  Lemma \ref{prec_versus_minus}, it follows $u<v$. Similarly, applying again Lemma \ref{prec_versus_minus}, from $\beta\prec \gamma$ we have  $v<u$, which is a contradiction. 
 \item   Suppose, for contradiction, that  $I_u, I_v\in {\mathcal I}_{Q}$ are such that $I_u\subsetneq I_v$ and  $I_u$ is neither a prefix nor a suffix of $I_v$.  In these hypotheses  there
must exist $\alpha, \alpha' \in I_v\setminus I_u$ and $\beta \in I_u$ such that  $\alpha \prec \beta \prec \alpha'$. Lemma \ref{prec_versus_minus}  implies
$v<u<v$, which is a contradiction. 
 \end{enumerate}
 Points $(3), (4)$ follow  similarly from Lemma  \ref{prec_versus_minus}. 
 \end{proof}
 
\begin{rem}\label{remark-complete} Clearly, Lemma \ref{convex_sets} given above continues to hold also in the case of \emph{complete} Wheeler automata, with  $ \text{\em Pref}(\mathcal L(\mathcal A)) $ replaced by  $ (\Sigma^{*},\prec) $. 
\end{rem}

Since Definition \ref{WNFA_WDFA} allows the transition function of  Wheeler DFA's to be \emph{incomplete}, one could wonder why not forcing completeness in the definition of Wheeler automaton. We can now show, using the above remark, that incompleteness is somehow necessary: the class of languages would be different if completeness were required.  

\begin{example} \label{incomplete}  {\bf A Wheeler language not  recognised by any  \emph{complete} WDFA}. 

Let  $\mathcal A=(Q, s, \delta,<,F)$ be the  following  (incomplete) WDFA such that $ \mathcal L(\mathcal A)=\mathcal L= b^{+}a $:
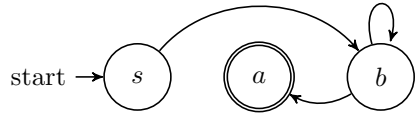
\begin{figure}[h!] 
 \begin{center}
\begin{tikzpicture}[->,>=stealth', semithick, auto, scale=.8]

 \node[state,initial, label=above:{}] (S)    at (-2,0)		{$s$};
 \node[state, accepting, label=above:{}] (U)    at (0,0)		{$a$};
 \node[state, label=above:{}] (V)    at (2,0)				{$b$};
 
\draw (S) edge [bend left=50, above] node [bend right=30, above] {} (V);
\draw (V) edge [bend left, above] node [bend right, above] {} (U);
\draw (V) edge  [loop above] node {} (V);
\end{tikzpicture} 
\end{center}
\end{figure}

Suppose, for contradiction,  that $\mathcal L=\mathcal L(\mathcal A')$, where $\mathcal A'=(Q', s', \delta',<',F')$ is a complete  Wheeler DFA.  
Since the set $Q'$ is finite, there exist $i,k\in \mathbb N$ with  $i<k$  and  $ {\delta}(s', b^i)= {\delta}(s', b^k)=u$, for some $u\in Q'$.  From  $b^ia\in \mathcal L$  it follows  $ {\delta}(s', b^ia)=z$ for some $z\in F'$. Consider now  $v\in Q'$ such that ${\delta}(s', ab^i)=v$.  
By  Remark \ref{remark-complete},  $I_u$ is an convex set in  the linear order consisting of all words read by Wheeler  automaton $ \mathcal A' $, ordered co-lexicographically, that is $(\Sigma^*, \preceq)$. Since  
$b^i \prec ab^i  \prec b^k$, and $b^i, b^k\in I_u$ 
implies  $ab^i \in I_u$ and since $\mathcal A'$ is a DFA,   $v=u$ follows.  But then ${\delta}(s', ab^ia)=z\in F'$ and  we would have $ab^ia\  \in  \mathcal L$,  contradicting  $\mathcal L=b^+a$. 
\end{example} 
 
From  Lemma \ref{convex+order} it  follows that   $(I_{Q}, \prec^{i})$  and  $(I_ {\text{\em Pref}(\mathcal L(\mathcal A))}, <^{i})$  are linear orders.   

\begin{lemma}\label{compare}\cite{alanko2020regular} Let  $\mathcal A=(Q, s, \delta,<,F)$ be a WNFA. Consider $ I_{u},I_{v}\in I_{Q} $ and $ I_{\alpha}, I_{\beta} \in I_ {\text{ Pref}(\mathcal L(\mathcal A))}$.  
\begin{enumerate}
\item $ I_u  \prec^{i} I_v $ implies that $ u < v $ and  $ u < v $ implies that  $ I_u  \preceq^{i} I_v $.
\item $ I_{\alpha} <^{i} I_{\beta} $ implies that $ \alpha \prec  \beta $ and  $ \alpha \prec  \beta $ implies that  $ I_{\alpha} \leq^{i} I_{\beta} $.
\item Any sequence of intervals $(I_{\alpha_{i}})_{i \geq 1}$ where $(\alpha_i)_{i\geq 1}$ is a monotone sequence in $(\pf L, \prec)$, is eventually constant. 
\end{enumerate}
\end{lemma}
\begin{proof} \

\begin{enumerate}
\item Suppose $ I_u  \prec^{i} I_v $. Then, either  there exists $ \alpha \in I_u$  such that for all $\beta\in I_v$ it holds $\alpha \prec \beta$, or   there exists $ \beta \in I_v$  such that for all $\alpha\in I_u$ it holds $\alpha \prec \beta$. In the first case, we have $\alpha \in I_u\setminus I_v$,  while in the second case we have $\beta\in I_v\setminus I_u$. In both cases $u<v$ follows from Lemma
 \ref{prec_versus_minus}. 
 
 For the second implication suppose,  for contradiction, that $u<v$  and $ I_v\prec^{i} I_u$ holds. Then, either there exists $\alpha \in I_v$ such that for all $\beta\in I_u$ it holds $\alpha\prec \beta$, or there exists $\beta\in I_u$ such that for all $\alpha \in I_v$  it holds $\alpha\prec \beta$. In the first case, $\alpha \in I_v \setminus I_u$, while in the second case $\beta\in I_u\setminus I_v$. In both case we obtain  $v<u$  by Lemma  \ref{prec_versus_minus}: a contradiction. 
\item This point is entirely similar to the above.
\item This is proved similarly to Corollary \ref{monotone} using (2), so we provide just a sketch: since $Q$ is finite, also the set of intervals on $Q$ is finite, thus by property (2) $(I_{\alpha_{i}})_{i \geq 1}$ must stabilize, being $(\alpha_i)_{i\geq 1}$ a monotone sequence in $(\pf L, \prec)$. 
\end{enumerate}
\end{proof}


If $\mathcal A$  is a   WNFA we can prove that the following construction, which is the ``convex version'' of the classic powerset construction for NFA, allows determinisation without exponential blow-up.

\begin{definition}
 If $\mathcal A=(Q, s, \delta, <,F)$  is a   WNFA  we define  its (Wheeler) \emph{determinization} as the automaton ${\mathcal A^{d}}=(Q^{d}, s^{d}, \delta^{d}, <,^{d}F^{d})$, where:
 \begin{itemize}
     \item[-] $Q^{d}=I_{\text{\em Pref}(\mathcal L(\mathcal A))}$;
     \item[-] $s^{d}= I_{\epsilon}=\{s\}$
     \item[-] $F^{d}=\{I_\alpha ~|~ \alpha \in \mathcal L(\mathcal A)\}$;
     \item[-] $\delta^{d}: I_{\text{\em Pref}(\mathcal L(\mathcal A))}\times \Sigma \rightarrow  I_{\text{\em Pref}(\mathcal L(\mathcal A))}$ is the partial function  defined as $\delta^{d}(I_\alpha,e)= I_{\alpha e}$, for all  $e\in \Sigma$ and $\alpha e\in \text{\em Pref}(\mathcal L(\mathcal A))$;
   \item[-] $<^{d}=<^{i}$. 
 \end{itemize} 
\end{definition}

\begin{lemma}[WNFA Determinization]\label{Wdeterminization}\cite{alanko2020regular}
	If $\mathcal A=(Q, s, \delta, <,F)$  is a   WNFA with $n$ states over an alphabet $\Sigma$ (with  at least one $a$-edge for each $a\in \Sigma$),   then $\mathcal A^{d}$ is a WDFA  with at most $2n-1-|\Sigma|$ states, and  $ \mathcal L(\mathcal A^{d}) = \mathcal L(\mathcal A)$. 
\end{lemma}
\begin{proof}
	The fact that $ \mathcal L(\mathcal A^{d}) = \mathcal L(\mathcal A)$ is seen as in the (classic) regular case: the reachable subsets of the powerset construction are exactly the ones in $ Q^{d} $ .   
	
	We  prove that $<^{d}$ is a Wheeler order on the states of the automaton $\mathcal A^{d}$.  By Lemma  \ref{convex_sets}, the set  $Q^{d}=I_{\text{\em Pref}(\mathcal L(\mathcal A))}$ of states of  $\mathcal A^{d}$ is a prefix/suffix family of intervals, so that,  by   Lemma \ref{convex+order},    $<^{d}$ is a linear order on $Q^{d}$.
	Next, we check the Wheeler properties.
	The only vertex with  in-degree $0$  is $I_\epsilon$, and it clearly precedes those with positive in-degree.
	For any two edges  $(I_\alpha,I_{\alpha a_1}, a_1)$,  $(I_\beta,I_{\beta a_2}, a_2)$ we have:
	\begin{itemize}
		\item[(i)]  if $a_1 \prec a_2$  then  $\alpha a_1\prec \beta a_{2}$,  and  from Lemma \ref{compare} it follows   $I_{\alpha a_1}\leq^d I_{\beta a_2}$.  Moreover, by the input consistency of  $\mathcal A$, states in $I_{\alpha a_1}$ are $a_1$-states, while states in $I_{\beta a_2}$ are $a_2$-states; hence   $I_{\alpha a_1}\neq I_{\beta a_2}$, so that $I_{\alpha a_1}<^d I_{\beta a_2}$ follows.
		\item[(ii)]   If $a=a_{1}=a_{2}$ and  $I_\alpha < I_\beta$,   from Lemma \ref{compare} it follows $\alpha\prec \beta$, so that $\alpha a \prec \beta a$ and, using again   Lemma \ref{compare}, we obtain   $I=I_{\alpha a}\leq^i I=I_{\beta a}$.
	\end{itemize}
	Finally, we prove that $|Q^d|\leq 2n-1-|\Sigma|$. By the Wheeler properties, we know that the only interval in $I_{\text{\em Pref}(\mathcal L(\mathcal A))}$  containing the initial state $s$ of the automaton $\mathcal A$ is $\{s\}$ and that the remaining intervals can be partitioned into $|\Sigma|$-classes, by looking at the letter labelling incoming edges. 
	Let $\Sigma=\{a_1, \ldots, a_k\}$, and, for every $i=1,\ldots,k$,  let  $m_i$  be  the number of states of   the automaton $\mathcal A$  whose incoming edges are labelled $a_i$: then  $\sum_{i=1}^k m_i=n-1$.  Using Lemma \ref{2n} we see that the intervals in $Q^d$ composed by  $a_i$ states are at most $2m_i-1$, so that the total number of intervals in $V^d$ is at most $1+ \sum_{i=1}^k (2m_i-1)=1+ 2( \sum_{i=1}m_i )-k= 1+2(n-1)-k=2n-1-k= 2n-1-|\Sigma|$. \hfill
\end{proof}

 We will use the following Lemma in the next section. 
 \begin{lemma}\cite{alanko2020regular} \label{sim_convex}  Let $\mathcal A=(Q, s, \delta, <,F)$ be a    WNFA,  $\alpha, \beta, \delta \in    \text{Pref}(\mathcal L(\mathcal A))$,  $u,v,w\in Q$. 
\begin{enumerate}
\item if $ \alpha\prec \delta\prec \beta$  and $I_\alpha=I_\beta$, then $I_\alpha= I_\delta $;
\item if $ u<w<v $ and $ I_{u}=I_{v} $, then $ I_{u}=I_{w} $.    
\end{enumerate} 
 \end{lemma}
 \begin{proof} \
 
 \begin{enumerate}
 \item Suppose $\alpha\prec \delta\prec \beta  ~\hbox{and}~ I_\alpha=I_\beta$. If $u\in I_\alpha=I_\beta$  then $\alpha, \beta \in I_u$ and since by Lemma
\ref{convex_sets}  $I_u$ is  a convex set,   $\delta\in I_u$ follows. Hence, $u\in I_\delta$, from which it follows that $I_\alpha\subseteq I_\delta$. Suppose, for contradiction, that $I_\alpha\subsetneq I_\delta$ and let $v\in I_\delta\setminus I_\alpha$. It follows  $\delta \in  I_v$ and $\alpha \not \in I_v$.  Consider  $u$ such that $\alpha \in I_u$,  then $\alpha \in I_u \setminus I_v$, $\delta\in I_v$, and $\alpha\prec \delta$, from which it follows that $u<v$ by Lemma \ref{prec_versus_minus}.   On the other hand, $\beta \in I_u\setminus I_v$ as well, because $u\in I_\beta=I_\alpha$ and  $v\not \in  I_\beta=I_\alpha$, then  $\delta \prec \beta$ and 
 Lemma \ref{prec_versus_minus} implies  $v<u$. A contradiction. 
 \item This point is entirely similar to the above. 
 \end{enumerate}
 \end{proof}

\subsection{ Convex Equivalences from  Wheeler Automata}

Given a WNFA  $\mathcal A$,  we consider two convex equivalence relations,  $\sim_{\mathcal A}$ and $ \approx_{\mathcal A}$.

 \begin{definition}\label{sim_A}If $\mathcal A=(Q, s, \delta, <,F)$ is an WNFA, $\alpha, \beta\in  \text{\em Pref}(\mathcal L(\mathcal A))$, and $u,v \in Q$,  we define:
\begin{align*}
\alpha\sim_{\mathcal A} \beta & \text{ if and only if } I_\alpha=I_\beta. \\ 
u \approx_{\mathcal A}  v & \text{ if and only if } I_{u}=I_{v}. 
\end{align*}
\end{definition}
Whe shall write $\approx$ instead of $\approx_{\mathcal A}$ when  the automaton $\mathcal A$ is clear from the context. Note that, by Lemma \ref{sim_convex},   $\approx $-equivalence classes  are  in fact intervals of $(Q,<)$---that is, $\approx $ is a convex equivalence over $(Q,<)$. As we shall see in Lemma \ref{sim_Aw}, the  equivalence $\sim_{\mathcal A}$   over $ \text{\em Pref}(\mathcal L(\mathcal A))$  is also convex, with respect to the co-lexicographic order on $ \text{\em Pref}(\mathcal L(\mathcal A))$.  

 \begin{definition}  \label{right_inv} Given a language $\mathcal L\subseteq \Sigma^*$, an equivalence relation $\sim$ over $\text{\em Pref}(\mathcal L)$    is:
\begin{itemize}
\item[-]   {\em right invariant}, when for all  $\alpha, \beta \in \text{\em Pref}(\mathcal L)$ and $\gamma\in \Sigma^*$: 
\begin{align*}
\text{ if } \alpha \sim \beta \text{ and } \alpha\gamma \in \text{\em Pref}(\mathcal L), &\text{ then }  \beta\gamma \in \text{\em Pref}(\mathcal L) \text{ and }  \alpha\gamma \sim \beta\gamma;
\end{align*}
\item[-]   {\em input consistent} if all  words belonging to the same $\sim$-class end with the same letter. 
\end{itemize}
\end{definition}

\begin{lemma}\label{sim_Aw}\cite{alanko2020regular} If $\mathcal A=(Q, s, \delta, <,F) $ is an $n$-states WNFA  such that $\mathcal L= \mathcal L(\mathcal A)$, then:
\begin{enumerate}
\item  $\sim_{\mathcal A}$ is a right invariant, input consistent,   convex   equivalence relation over $\text{Pref}(\mathcal L)$;
\item  $\sim_{\mathcal A}$'s index  is less than or equal to $2n-1-|\Sigma|$;
\item $\mathcal L$  is a  union of $ \sim_{\mathcal A}$-classes. 
\end{enumerate}
\end{lemma}

\begin{proof} \
\begin{enumerate}
    \item 
We first  check that $ \sim_{\mathcal A}$ equivalence classes are convex sets (convex sets) of  $(\text{\em Pref}(\mathcal L), \prec)$.
If $\alpha\prec \beta\prec \gamma$ are such that $\alpha, \beta, \gamma\in \text{\em Pref}(\mathcal L)$ and $\alpha\sim_{\mathcal A} \gamma$,
then $\beta \sim_{\mathcal A} \alpha$ follows from Lemma \ref{sim_convex}. 

As for right invariance, suppose $\alpha \sim_{\mathcal A} \beta$. Then $I_\alpha=I_\beta$, from which it follows $I_{\alpha e}=I_{\beta e}$ because for any state $u\in I_{\alpha e} $     there exists a state $u'\in I_\alpha=I_\beta$  such that  $u'\in \delta(u,e)$; hence $u\in I_{\beta e} $. This proves that  $I_{\alpha e}\subseteq I_{\beta e}$. The reverse inclusion is proved similarly. 

Input consistency of $ \sim_{\mathcal A} $  follows from Wheeler properties, since if two words end with  different letters, then   they cannot lead to the same state in a Wheeler automaton. 

\item 
The index of $ \sim_{\mathcal A} $ is equal to the cardinality of     $I_ {\text{\em Pref}(\mathcal L(\mathcal A))}$  which is a prexix/suffix family of  $(Q,<)$ by Lemma \ref{convex_sets}.  By Lemma\ref{2n},  this index is bounded by  $2n-1-|\Sigma|$.
\item 
$\mathcal L=\bigcup_{\alpha\in \mathcal L} [\alpha]_{ \sim_{\mathcal A}}$.
\end{enumerate}
\end{proof}

If  $\mathcal A$ is a WDFA, $\mathcal L= \mathcal L(\mathcal A)$, and $\alpha \in  \text{\em Pref}(\mathcal L)$, then $I_\alpha$  contains a single state: $\sim_{\mathcal A}$'s index is equal to the number of states of the automaton ${\mathcal A}$. 

\medskip

Let us now consider the second equivalence, $ \approx_{\mathcal A}$ (or, simply, $\approx$). 

\begin{definition}\label{def:A/equiv} Let $\mathcal A=(Q, s, \delta, <,F)$ be a WNFA.  The 
quotient automaton ${\mathcal A}/{\approx }=(Q^{\approx }, s^{\approx },\delta^{\approx },<^{\approx },F^{\approx })$  is defined as follows:

\begin{itemize}
\item[-] $Q^{\approx }=\{[u]_{\approx }~|~ u \in Q\}$;
\item[-] $s^{\approx }=[s]_{\approx }= \{s\}$;
\item[-] $\delta^{\approx }([v]_{\approx },e)=\{[u]_{\approx } ~|~(\exists u' \in [u]_{\approx })(\exists v'\in [v]_{\approx }) (u'\in \delta(v',e))\}$;
\item[-] $[u]_{\approx }<^{\approx }[v]_{\approx }$ if and only if $ [u]_{\approx } \neq [v]_{\approx } \land u<v$;
\item[-] $F^{\approx }=\{  [u]_{\approx }~|~ [u]_{\approx }\cap F\neq \emptyset\}$. 
\end{itemize}
\end{definition}

 Note that the  relation $<^\approx$ on the equivalence classes is well defined because, by Lemma \ref{sim_convex},  the equivalence classes $[u]_{\approx }$ are (disjoint) intervals  of  $(Q,<)$.


\begin{lemma}\label{quotient} ${\mathcal A}/{\approx }$  is a Wheeler automaton and  $\mathcal L({\mathcal A})= \mathcal L({\mathcal A}/{\approx })$.
\end{lemma}
\begin{proof}
The fact that the order  on equivalence classes defined above  is Wheeler follows easily from the definition and the fact that the equivalence classes are intervals.

To see that $\mathcal L({\mathcal A})= \mathcal L({\mathcal A}/{\approx })$,  observe that, although  in general  the implication 
$ [u]_{\approx }  \in \delta^{\approx }([v]_{\approx }, e) \Rightarrow u\in \delta(v,e) $ does not hold,  we  do have that  
$[u]_{\approx }  \in \delta^{\approx }(s^{\approx },e) \Rightarrow  u\in \delta(s,e)$ does hold. As a matter of fact, more generally, we can prove that  for all $\alpha \in \Sigma^*$: 
\begin{align}\label{eq1}
[u]_{\approx }   \in {\delta}^{\approx }(s^{\approx },\alpha) & \text{ if and only if }  u\in {\delta}(s,\alpha). 
\end{align} 

The direction from left to right of (\ref{eq1}) is proved by induction on $|\alpha|$. 

For the base case, suppose $[u]_{\approx }  \in {\delta}^{\approx }(s^{\approx },\epsilon)=s^{\approx } $;  then, since 
$s^{\approx }=\{s\}$ we have $u=s\in  \delta(s,\epsilon)$.

 For the inductive step, suppose $[u]_{\approx }  \in {\delta}^{\approx }(s^{\approx },\alpha e)$; then  let  $v\in Q$ be such that  
$ [ v]_{\approx }  \in {\delta}^{\approx } (s^{\approx },\alpha)$,  and $[u]_{\approx }   \in \delta^{\approx } ( [v]_{\approx }, e)$.
By inductive hypothesis,  $v\in {\delta}(s, \alpha)$ and by definition of $\delta^{\approx } $ we know that 
  there are $   u', v'\in Q$, such that $ u {\approx }  {u'},  v{\approx }   {v'} $, and  $u'\in \delta(v',e)$.  From $[v]_{\approx }=[ v']_{\approx }$ it  follows that $v'\in  \delta(s, \alpha)$,  and so  $u'\in  \delta (s, \alpha e)$. Since  $[u]_{\approx }=[ u']_{\approx }$, it  follows   $u\in   \delta (s, \alpha e)$.  
  
The direction from right to left of (\ref{eq1}) is easy to see.

\medskip
  
  From (\ref{eq1}) $ {\mathcal L}(\mathcal A)= \mathcal L({\mathcal A}/{\approx })$ follows.  In fact:
  $ \alpha \in  {\mathcal L}(\mathcal A) \Leftrightarrow   (\exists u \in F) ( u\in \delta(s, \alpha))  \Leftrightarrow  (\exists u \in F)([u ]_{\approx }  \in \delta(s^{\approx }, \alpha))
    \Leftrightarrow (\exists [u]_{\approx }  \in F^{\approx })   ([u ]_{\approx }  \in \delta(s^{\approx }, \alpha)) \Leftrightarrow  \alpha \in  \mathcal L({\mathcal A}/{\approx })$.

\end{proof}
When the $\approx$-classes are not singletons, two different  states in a (W)NFA can be reached by exactly the same collection of $ \alpha $'s in $ \Sigma^{*} $. To avoid this trivial kind of redundancy, we introduce the following notion.

\begin{definition}\label{red}
 A Wheeler NFA  $\mathcal A=(Q, s,\delta,<,F)$  is {\em  reduced} if for all $u,v\in Q$, 
\begin{align*}
u\neq v \text{ if and only if } I_u\neq I_v.
\end{align*}
\end{definition}

 It is clear that the quotient automaton ${\mathcal A}/{\approx }$  of a WNFA is reduced.
As a consequence of Lemma \ref{quotient} we have:
\begin{corollary} Any  WNFA is  equivalent to a reduced one.
\end{corollary}

Our interest in reduced automata relies on the following result:

\begin{lemma} \label{unique}\cite{alanko2020regular} The Wheeler order of a reduced WNFA  is unique.
\end{lemma} 
\begin{proof} Let   $\mathcal A=(Q, s,\delta,<,F)$ be a reduced  WNFA. 
If $u\neq v\in Q$ then $I_u\neq I_v$ and either $I_u\setminus I_v\neq   \emptyset$ or 
$I_v\setminus I_u\neq   \emptyset$.  If  $I_u\setminus I_v\neq   \emptyset$, consider $\alpha\in I_u\setminus I_v$, and $\beta \in I_v$.
Then, by Lemma \ref{prec_versus_minus}, if $\alpha \prec \beta$ then  $u <v$.  
Similarly, if $\alpha\in I_v\setminus I_u$ and  $\beta \in I_u$, we have that $\alpha \prec \beta$ implies $v<u$.
 
In both cases, the  Wheeler order is (uniquely) determined. 
\end{proof}


  In Corollary   \ref{reducedwnfa} we shall see that  deciding whether a given Wheeler NFA is reduced is in $P$.    Reduced NFA are considered  again in Section \ref{AWheeler}, 
where we prove that deciding Wheelerness for a reduced NFA can be done in polynomial time (contrary to  the case of  general NFA, see \cite{DBLP:conf/esa/GibneyT19}).

 \subsection{A Myhill-Nerode Theorem for Wheeler Languages }\label{MN-subsection}

Given  $\mathcal L\subseteq \Sigma^*$,   we define the \emph{right context} of $\alpha\in \Sigma$, as
\[ \alpha^{-1}\mathcal L =\{\gamma\in  \Sigma^*~:~ \alpha\gamma \in \mathcal L\},\]

and we denote by 
$\equiv_{\mathcal L}$ the Myhill-Nerode equivalence (\emph{right syntactic congruence}) on $ \text{\em Pref}(\mathcal L)$ defined as 
 \[\alpha \equiv_{\mathcal L}\beta ~~\Leftrightarrow ~~ \alpha^{-1}\mathcal L=\beta^{-1}\mathcal L.\]

\begin{definition}
The input consistent, convex refinement $\equiv_{\mathcal L}^{c}$ of $ \equiv_{\mathcal L} $ is  
defined  as follows:
\begin{align*}
\alpha\equiv_{\mathcal L}^{c}\beta  \Leftrightarrow &  \alpha \equiv_{\mathcal L} \beta \wedge end(\alpha)=end(\beta) \wedge (\forall    \gamma \in  \text{\em  Pref}(\mathcal L))  (min\{\alpha,\beta\}\prec \gamma \prec max\{\alpha, \beta\} \rightarrow \gamma \equiv_{\mathcal L} \alpha),
\end{align*}
where  $\alpha, \beta \in  \text{\em Pref}(\mathcal L)$ and  $end(\alpha)$ is the final character of $\alpha$ when $\alpha\neq\epsilon$, and $ \epsilon $ otherwise. 
\end{definition}

\begin{lemma}\label{base} \cite{alanko2020regular}
If $\mathcal L\subseteq \Sigma^*$,  then $\equiv_{\mathcal L}^{c}$    is a   convex, right invariant, input consistent equivalence relation over $(\text{Pref}(\mathcal L), \prec)$  and $\mathcal L$ is a union of classes of  $\equiv_{\mathcal L}^{c}$.
 \end{lemma}
 \begin{proof} 
 The equivalence $\equiv_{\mathcal L}^{c}$ is input consistent by definition.  
 Moreover, it  is  convex, being a convex refinement of an equivalence over  the  ordered set
  $ (\text{\em Pref}(\mathcal L), \prec)$  (see Lemma \ref{wheeler_refinement}). 

To prove that $\equiv_{\mathcal L}^{c}$ is right invariant, consider  $\alpha , \alpha' , \gamma\in  \text{\em Pref}(\mathcal L)$ and assume $\alpha\equiv_{\mathcal L}^{c} \alpha'$.   Note that:
\begin{itemize} 
\item[-] if $\alpha \gamma \in \text{\em Pref}(\mathcal L)$ then  there exists $\nu \in \Sigma^*$ such that  $\alpha \gamma \nu \in  \mathcal L$, therefore  $\alpha' \gamma \in \text{\em Pref}(\mathcal L)$ follows from $\alpha\equiv_{\mathcal L} \alpha'$;
\item[-]  $\alpha \gamma \equiv_{\mathcal L} \alpha'\gamma$ follows from $\alpha\equiv_{\mathcal L} \alpha'$. 
\item[-] If $\alpha \gamma \prec \beta' \prec \alpha' \gamma$,  for $\beta' \in  \text{\em Pref}(\mathcal L)$, then $\beta'=\beta \gamma$,    and $\alpha \prec \beta \prec \alpha'$. Since $\alpha, \alpha'$ belong to the same $\equiv_{\mathcal L}^{c} $ class, then     $\beta\equiv_{\mathcal L} \alpha$, and $\beta'=\beta \gamma \equiv_{\mathcal L} \alpha \gamma$ follows. 
\end{itemize}
Since $\alpha \gamma, \beta \gamma $ end with the same letter, the previous points imply  that  $\equiv_{\mathcal L}^{c} $ is right invariant. 

Finally, $\mathcal L$ is a union of classes of  $\equiv_{\mathcal L}^{c}$ because $\mathcal L$   is a union of $ \equiv_{\mathcal L}$ classes and $\equiv_{\mathcal L}^{c}$ is a refinement of $ \equiv_{\mathcal L}$.
\end{proof}

\begin{lemma}\label{finite-lemma}\cite{alanko2020regular}
 If $ \mathcal A=(Q, s, \delta, <,F) $ is a WNFA and $\mathcal L = \mathcal L(\mathcal A)$,  then  $\sim_{\mathcal A}$ is a refinement of $\equiv_{\mathcal L}^{c}$.   
\end{lemma}
\begin{proof}
Suppose $\alpha \sim_{\mathcal A} \beta$;  then $\alpha \equiv_{\mathcal L} \beta$ follows easily from the definition of  $ \sim_{\mathcal A}$,  and $end(\alpha)=end(\beta)$ follows  from the input consistency of    $\mathcal A$.  
To prove that $\alpha \equiv_{\mathcal L}^{c} \beta$ we only have to show that if $\gamma \in \text{\em Pref}(\mathcal L)$ and $min\{\alpha,\beta\}\prec \gamma \prec max\{\alpha, \beta\}$ then $\gamma \equiv_{\mathcal L} \alpha$. This holds because, by Lemma \ref{sim_convex}, from  $\alpha\prec \gamma \prec \beta$ and $I_{\alpha}=I_\beta$, we have 
$I_\alpha=I_\gamma$, hence 
$\alpha \sim_{\mathcal A} \gamma$  holds, and   $\alpha \equiv_{\mathcal L} \gamma$ follows. 
\end{proof}

\begin{corollary}\label{finite-corollary} \cite{alanko2020regular}
If $ \mathcal A=(Q, s, \delta, <,F) $ is a WNFA with $|Q|=n$  and $\mathcal L = \mathcal L(\mathcal A)$,   then   $\equiv_{\mathcal L}^{c}$'s index  is  bounded by  $2n-1-|\Sigma|$. 
\end{corollary}
\begin{proof}  By Lemma \ref{finite-lemma}, we know that 
  $\sim_{\mathcal A}$ is a refinement of $\equiv_{\mathcal L}^{c}$, hence the number of classes of $ \equiv_{\mathcal L}^{c}$ is less than or equal to the number of classes of $\sim_{\mathcal A}$, which is bounded by  $2n-1-|\Sigma|$, as proved in the Lemma \ref{sim_Aw}. 
    \end{proof} 
 Note that, if $\mathcal L$ is Wheeler, we cannot always  extend $\equiv_{\mathcal L}^{c}$ to the set $\Sigma^*$ maintaining the preceding corollary.  For example, if $\mathcal L$ is the Wheeler language of Example \ref{incomplete}, then the equivalence relation $\equiv_{\mathcal L}^{c}$ has an infinite number of classes over $\Sigma^*$.

\begin{theorem}[Myhill-Nerode for Wheeler Languages]  \label{Myhill-Nerode} \cite{alanko2020regular}
Given a language $\mathcal L\subseteq \Sigma^*$, the  following are equivalent:
\begin{enumerate}
\item $\mathcal L$ is a Wheeler language (i.e. $L$  is recognized by a  WNFA).
\item $\equiv_{\mathcal L}^{c}$ has  finite index.
\item $\mathcal L$ is a union of   classes of a convex, input consistent, right invariant  equivalence over  $(\text{Pref}(\mathcal L), \prec)$  of finite index.
\item $\mathcal L$ is recognized by a  WDFA.
\end{enumerate}
\end{theorem}

\begin{proof} \

\begin{itemize}
  \item[] (1) $\Rightarrow$ (2) From Corollary \ref{finite-corollary}.
  \item[] (2) $\Rightarrow$ (3)   $\mathcal L$ is a union of $\equiv_{\mathcal L}^{c}$ classes,   which  by Lemma \ref{base}, is a convex, input consistent, right invariant  equivalence of  of finite index. 
\item[] (3) $\Rightarrow$ (4)   Suppose   $\mathcal L$ is a union of classes of  a convex, input consistent,  right invariant  equivalence  relation $ \sim $  of finite index. We build a WDFA    ${\mathcal A}_\sim=(Q_{\sim}, s_{\sim}, \delta_{\sim},<_{\sim} F_{\sim})$ such that $\mathcal L=\mathcal L(\mathcal A)$ as follows: 
\begin{itemize}
\item[-] $ Q_{\sim}=\{[\alpha]_{\sim}~|~\alpha \in  \text{\em Pref}(\mathcal L)\}$;
\item[-] $ s_{\sim}=\{[\epsilon]_{\sim}\} $ (note that, by input consistency,  $[\epsilon]_{\sim}=\{\epsilon\}$);  
\item[-]  if $ Ie \cap \text{\em Pref}(\mathcal L) \neq \emptyset $ and $ Ie \subseteq J $, then $ \delta_{\sim}(I,e) = J $  (note that $J$, if existing, is unique because equivalence classes are pairwise disjoint);
\item[-] $<_{\sim} = \prec^{i}$,  that is: $I<_{\sim} J$ if and only if $(\forall \alpha \in I) (\forall \beta \in J) ~\alpha \prec \beta$.
\item[-] $ F_{\sim}= \{I~:~ I \subseteq \mathcal L\}$.    
\end{itemize}
For all $I\in Q_{\sim}$ and $\alpha \in \text{\em Pref}(\mathcal L)$,  observe that $ \delta_{\sim}(I,\alpha) $  (if defined) is always a singleton set (i.e. $ \mathcal A_{\sim} $ is deterministic).  

We prove that:
\begin{align*}
\alpha\in I & \text{ if and only if }   \delta_{\sim} (s_{\sim}, \alpha)=I.
\end{align*}
  by induction on the length of  $\alpha\in  \text{\em Pref}(\mathcal L)$. 
  If  $\alpha=\epsilon$ then  $ \delta_{\sim}(s_{\sim}, \alpha)=[\epsilon]_{\sim}$ and  $[\epsilon]_{\sim}=\{\epsilon\}$, by definition. 
If $\alpha=\alpha' e\in  \text{\em Pref}(\mathcal L)$ with  $e\in \Sigma$,  then     $\alpha'\in  \text{\em Pref}(\mathcal L)$ and 
\[\alpha'e\in I \Leftrightarrow \exists J (\alpha'\in J \wedge \emptyset \neq Je \subseteq I )\Leftrightarrow  \exists J ( \delta(s_\sim, \alpha')=J \wedge \emptyset \neq Je \subseteq I \Leftrightarrow   \delta(s_\sim, \alpha)=I. \]    

From the above claim and the definition of $ F_{\sim} $, it easily follows that $\mathcal L$ is the language recognised by  ${\mathcal A}_{\sim}$.

\medskip 

We conclude by checking that  ${\mathcal A}_{\sim}$ is Wheeler, proving the two Wheeler properties (i) and (ii).

 To see Wheeler-(i) assume  $e\prec e' $ with  $e,e'\in\Sigma$. Consider  $I,J \in Q_{\sim}$ such that both  $\delta_{\sim}(I, e)$ and $\delta_{\sim}(J,e')$ are defined and are equal to 
 $H$ and $K$, respectively.  By definition of $\delta_{\sim}$,  there are  $\alpha\in I$, $ \alpha'\in J$ with  $\alpha e \in H$ and $ \alpha' e' \in K$. From  $e\prec e'$ it follows that  $H \prec^{i} K$ since all words in $H$ end with $e$, while all words in $K$ end with $e'$.

To see Wheeler-(ii) assume  $I<_{\sim} J$,  $e\in \Sigma$,  and both $\delta_{\sim}(I,e)$ and  $\delta_{\sim}(J,e)$ are defined and equal to $ H$ and  $K$, respectively. In these hypotheses there are  $\alpha\in I$,  $\alpha'\in J $, with  $\alpha e \in H$ and $ \alpha' e \in K$.  It follows $\alpha \prec \alpha'$ and therefore, $\alpha e \prec \alpha' e$ and   $ H \preceq^{i} K $. 
  
 This ends the proof of the implication $(3) \Rightarrow (4)$. 

\medskip

\item[](4) $\Rightarrow$ (1) Trivial. 
\end{itemize}  
\end{proof}

\begin{rem} \label{duality}
 If  $ \mathcal D$  is a WDFA  with $|Q|=n$ states,     then the  equivalence  $ \sim_{\mathcal D}$  over  $\text{\em Pref}(\mathcal L)$ defined in  Def. \ref{sim_A} has  $n$ classes, because each class $[\alpha]_{ \sim_{\mathcal D}}$ can be uniquely identified with the unique state $u_\alpha= {\delta}(s, \alpha)$. 
Moreover,   $ \sim_{\mathcal D}$ is a   convex, input consistent, right invariant  equivalence (Lemma \ref{sim_Aw}) and we may   construct    the  WDFA  
${\mathcal A}_{ \sim_{\mathcal D}}$ described in $(3\Rightarrow 4)$ of 
Theorem  \ref{Myhill-Nerode}; note that  ${\mathcal A}_{ \sim_{\mathcal D}}$ is  isomorphic to  $ \mathcal D$,  via the map 
$\phi: Q_{ \sim_{\mathcal D}} \rightarrow Q$ defined as $\phi([\alpha]_{ \sim_{\mathcal D}})=u_\alpha$.
\end{rem}

 \begin{corollary}\label{wdeterminization}   \cite{alanko2020regular}
If $ \mathcal A=(Q, s, \delta, <,F) $ is a WNFA with $|Q|=n$  and $\mathcal L = \mathcal L(\mathcal A)$,   then   there exists a unique, minimum-size (on the number of states) WDFA $ \mathcal B $
 such that $\mathcal L= \mathcal L(\mathcal B)$ and the number of $ \mathcal B $'s states is less   than  or equal to $2n-1-|\Sigma|$.
Moreover, the construction of $\mathcal B$  is effective and can be done in polynomial time. \end{corollary}
\begin{proof}
If $\mathcal L$ is recognized by an $ n $-states  WNFA,  then from Lemma \ref{base}  we know that  the  equivalence  $\equiv_{\mathcal L}^{c}$  is a  convex, input consistent, right invariant  equivalence  relation of  finite index, and $\mathcal L$ is a union of its classes.  Hence, using the construction employed to prove (3) $ \Rightarrow $ (4) of Theorem \ref{Myhill-Nerode}, we can build a WDFA $\mathcal B= \mathcal A_{\equiv_{\mathcal L}^{c}} $, whose number of states is equal to the number of $ \equiv_{\mathcal L}^{c} $-classes. From Corollary \ref{finite-corollary}  we know that the number of classes of $\equiv_{\mathcal L}^{c}$ is bounded  by   $2n-1-|\Sigma|$.

To see that $ \mathcal B $ has the minimum number of classes observe that, by Lemma \ref{finite-lemma}, any automaton $ \mathcal D $ accepting $ \mathcal L $ induces an equivalence relation $ \sim_{\mathcal D} $ which is a refinement of $ \equiv_{\mathcal L}^{c} $. If $ \mathcal D $ is deterministic, the number of $ \sim_{\mathcal D} $-classes is equal to the number of $ \mathcal D $'s states  which  is, therefore,  greater or equal than the number of  
 $ \equiv_{\mathcal L}^{c} $-classes.  It follows that,  if $ \mathcal D $  is a WDFA with the minimum number of states among
     WDFA's recognising $\mathcal L$, then $  \sim_{\mathcal D}  ~=   \equiv_{\mathcal L}^{c}$;
     this implies that $\mathcal A_{\equiv_{\mathcal L}^{c}}=  \mathcal A_{\sim_{\mathcal D}}$ so that 
  \[\mathcal B = \mathcal A_{\equiv_{\mathcal L}^{c}}=  \mathcal A_{\sim_{\mathcal D}}\simeq   \mathcal D. \]
  where the last isomorphism follows from  Remark \ref{duality}. 
 For the effectiveness of the construction of $\mathcal B$ we refer to \cite{alanko2020regular}.\end{proof}

  \begin{corollary}  \label{reducedwnfa} We can decide in polynomial time whether  a  Wheeler  NFA is reduced.
 \end{corollary}
 \begin{proof} Let $\mathcal A$ be a Wheeler NFA. For each pair $u,v$ of  $\mathcal A$-states, we consider the two Wheeler automata $\mathcal A^u, \mathcal A^v$ which are obtained from $\mathcal A$ by considering, as set  of final states, $\{u\}$, $\{v\}$, respectively. Note that we can test in polinomial time whether $\mathcal L(\mathcal A^u)= \mathcal L(\mathcal A^v)$, because, by Corollary  \ref{wdeterminization},   we can determinize  WNFA   in polynomial time, and check if their languages are equal still in polynomial time (since they are deterministic). Then, $\mathcal A$ is reduced iff $\mathcal L(\mathcal A^u) \neq \mathcal L(\mathcal A^v)$ for all pairs $u\neq v$.
 \end{proof}

\section{Testing Wheelerness}
\subsection{Is a $ \mathcal L$  Wheeler?}\label{Lwheeler?}

In this section we prove that, given a  regular language $\mathcal L$ (say, by  an NFA $\mathcal A$ recognizing it), it s decidable whether or not $\mathcal L$ is Wheeler. 
Moreover,  if we start from a DFA recognizing $\mathcal L$,  we describe a polynomial time algorithm to complete the task. Note that in this section we deal with standard edge-labeled automata. 

We begin by giving an automata-free characterization of Wheelerness. 

\begin{lemma} \label{infinite_nonwh} A  regular language $\mathcal L$   is    Wheeler if and only if     all monotone sequences in $(\pf L, \prec)$  become eventually constant modulo $\equiv_{\mathcal L}$. In other words,  for all sequences    $(\alpha_i)_{i\geq 1}$  in  $  \pf L$ with 
\begin{align*}
\alpha_1\preceq  \alpha_2\preceq \ldots\preceq \alpha_i  \preceq \ldots &\text{ or } 
\alpha_1\succeq \alpha_2\succeq\ldots\succeq \alpha_i  \succeq \ldots 
\end{align*}
 there exists an $n$ such that $\alpha_h\equiv_{\mathcal L}\alpha_k$,  for all $h,k\geq n$. 
\end{lemma}

\begin{proof} 
For the direction from left to right, suppose that  $\mathcal L$ is Wheeler and  consider an  infinite monotone  sequence  $(\alpha_i)_{i\geq 1}$ in $(\text{\em Pref}(\mathcal L), \prec)$.  
 By Theorem \ref{Myhill-Nerode}   there exist a WDFA  $\mathcal A=(Q,q,\delta, F, <)$ recognizing $\mathcal L$ and from    Corollary \ref{monotone}  it follows that there exists $n$ such that $\delta(s, \alpha_h)=\delta(s, \alpha_k)$, for all $k,h\geq n$. This, in turn, implies that  $\alpha_h\equiv_{\mathcal L}\alpha_k$,  for all $h,k\geq n$.

\medskip

For the direction from right to left, suppose the regular language $\mathcal L$   is not  Wheeler. By   Theorem \ref{Myhill-Nerode}   we know that $\equiv_{\mathcal L}^c$ has infinite 
index. However, since $\mathcal L$ is regular,  the equivalence  $\equiv_{\mathcal L} $  has finite index;  hence there exists a sequence $(\gamma_i )_{i\geq 1}$  of elements which are equivalent with respect to   $\equiv_{\mathcal L} $ but  pairwise 
not $\equiv_{\mathcal L}^c$-equivalent. From this sequence one can easily   extract a subsequence $(\beta_i)_{i\geq 1}$ which is either monotone increasing or 
monotone decreasing  and composed  of  $\equiv_{\mathcal L} $-equivalent elements (either   the set $\{i\geq 1: \forall   j>i ~(\gamma_j\prec\gamma_i)\}$ is finite, and we extract an infinite
increasing  subsequence, or  is infinite and we extract an infinite   decreasing sequence). 
Suppose the sequence $(\beta_i)_{i\geq 1}$  is decreasing (a similar argument can be used in case it is increasing).  By possibly discarding a finite number of initial elements from such a  sequence, we may assume  that all $\beta_i$'s end with the same letter. Then, for all $i$,  from
  $\beta_i\not \equiv_{\mathcal L}^c \beta_{i+1}$ and   $\beta_i \equiv_{\mathcal L} \beta_{i+1}$ it follows that   there exits $\eta_i\in \text{\em Pref}(\mathcal L)$ such that: 
  \[ \beta_i\succ \eta_i \succ \beta_{i+1}\] and $\beta_i \not \equiv_{\mathcal L} \eta_i$.
If we define   $(\alpha_i)_{i\geq 1}=(\beta_1,\eta_1, \beta_2, \eta_2, \ldots   )$, then $(\alpha_i)_{i\geq 1}$ is monotone  in  $  (\text{\em Pref}(\mathcal L), \prec)$, but there exists no $n$ such that $\alpha_h\equiv_{\mathcal L}\alpha_k$,  for all $h,k\geq n$. 

\end{proof}

%
   \begin{example} \label{nw} If  $\Sigma=\{a\}$,   we see that the regular language $\{a^{2i+1}: i\geq 0\}$ is not Wheeler by considering  the sequence $(\alpha_i)_{i\geq 1}$ with $\alpha_i=a^i$.  
 Another example of application of Lemma \ref{infinite_nonwh} is the language $\mathcal L=ax^*b ~|~ cx^* d$ which was proved to be non Wheeler in \cite{gagie2017wheeler}. 
 Consider the sequence 
 \[\alpha_i=
 \begin{cases}
 ax^i ~\text{ if $i$ is odd};\\
  cx^i~\text{ if $i$ is even}
 \end{cases}
 \]
 Then $(\alpha_i)_{i\geq 1}$ is a monotone (increasing) sequence in $(\pf L, \preceq)$ with $\alpha_i\not \equiv_{\mathcal L}\alpha_{i+1}$, and from Lemma \ref{infinite_nonwh}  it follows that $\mathcal L$ is not Wheeler.
  \end{example}

\begin{rem}\label{co-lex-pro}
In the following theorem we shall use some simple  properties  of the co-lexicographic order: 
 \begin{align}
& \xi \prec \zeta   \Leftrightarrow  \xi\rho \prec \zeta \rho,&\\
& \xi \prec \zeta  \Rightarrow  \xi \prec \rho \zeta, &\\
 & |\zeta| \geq |\xi|  \wedge \xi \succ \zeta  \Rightarrow   \xi \succ \rho \zeta,&
\end{align}
\end{rem}

\begin{theorem} \label{decidability} Consider a   regular language $\mathcal L=\mathcal L(\mathcal A)$, where   $\mathcal A$ is  the  minimum edge-labeled DFA recognizing   $\mathcal L$ with   initial state $s$. Then  $\mathcal L$ is not Wheeler if and only if there exist strings $\mu, \nu$, and  $\gamma$  such that:
\begin{enumerate} 
\item $\mu$ and $\nu$  label  paths from $s$  to states $u$ and $  v$, respectively,  with $u\neq v$;
\item  $\gamma$  labels  two   cycles, one  starting from $u$ and one starting from $v$;
\item $\mu, \nu\prec \gamma$ or  $\gamma \prec  \mu,  \nu$;
 \item   $|\mu|, |\nu|<|\gamma|\leq 2+|\mathcal A|+2|\mathcal A|^2+|\mathcal A|^3$,    where $|\mathcal A|$ is the number of states of the automaton $\mathcal A$.
\end{enumerate}
\end{theorem}
 
\begin{proof} We first prove that the four conditions above are sufficient to prove $ \mathcal L $ is not Wheeler. If $\mu, \nu$, and $ \gamma$ are as above,  then $\mu\neq \nu$ since they end in distinct states $u,v$ and $ \mathcal A $ is deterministic. 

Suppose now, without loss of generality, that $\mu \prec \nu$ and:
\begin{enumerate}
\item[a)] \noindent  if  $\mu\prec  \nu \prec \gamma$,    let  $ \eta_i=\mu\gamma^i, ~~ \beta_i =\nu \gamma^i$, while
\item[b)] if $\gamma \prec \mu\prec \nu $,   let  $\eta_i=\nu \gamma^i, ~~ \beta_i =\mu\gamma^i$. 
\end{enumerate}

Note that, in both cases,   all $\eta_i$'s and $\beta_i$'s belong to $\pf L$ and $\eta_i\not \equiv_{\mathcal L}\beta_i$, because  $\eta_i, \beta_i$ end  in different nodes  $u,v$ of the minimum automaton.  Moreover,  for any $ i $,  $ \eta_{i} \prec \beta_{i} $, in  case a) while $ \eta_{i} \succ \beta_{i} $ in case b).  Finally, it can easily be checked that $\beta_i\prec \eta_{i+1}$ holds in case a), while
 $\beta_i\succ\eta_{i+1}$ holds in case b) since $\mu\succ \gamma $ and $\gamma$ is not a suffix of $\mu$, being $|\gamma|>|\mu|$. 
 
 Hence we have: 
\begin{enumerate}
\item[a)] $ \eta_1\prec \beta_1\prec \ldots\prec \eta_i\prec \beta_i \prec \ldots  $
\item[b)]  $ \eta_1\succ\beta_1\succ \ldots\succ \eta_i\succ \beta_i \succ \ldots $
\end{enumerate}
In  both cases we have a  monotone sequence in $\pf L$ which is not eventually constant modulo $\equiv_{\mathcal L}$,  so that  $\mathcal  L$ is not Wheeler by Lemma \ref{infinite_nonwh}. 
\medskip

We now prove the converse of our main statement: if $\mathcal L$ is not Wheeler we can find $ \mu, \nu$ and $   \gamma$ satisfying conditions (1)-(4) above.

\begin{claim}
\label{firstclaim}
If $\mathcal L$ is not Wheeler, 
there exist  words $\alpha, \beta, \alpha', \gamma'\in \pff L$ such that:
\begin{itemize}
\item [-]$\alpha\prec \beta\prec \alpha'$;
\item  [-] $\alpha,\alpha'$ end in a state $u$ and $\beta$ ends in $v$ with $u\neq v$;
\item [-]   $|\gamma'|\leq |\mathcal A|^2$ and $\gamma'$ labels two cycles starting from $u$ and $v$, respectively;
\item[-] $|\alpha|, |\beta|, |\alpha'|\leq 2+ |\mathcal A|+|\mathcal A|^2+|\mathcal A|^3$.
\end{itemize}
\end{claim}

To prove the above claim we apply   Lemma \ref{infinite_nonwh}. Consider  a  monotone  sequence   $(\alpha_i)_{i\geq 1}$  which is not eventually constant modulo $\equiv_{\mathcal L}$. Assume that $\alpha_i\prec \alpha_{i+1}$, for all $i$ (the case $\alpha_i\succ \alpha_{i+1}$, for all $i$,  is analogous).   
By possibly  erasing  a finite number of initial elements in the sequence,  we can   assume that all $\alpha_i$'s  end with the same $|\mathcal A|^2+1$ letters (this is possible by the finiteness of $\Sigma$ and by the fact that the monotone sequence $(\alpha_i)_{i\geq 1}$ is not eventually constant). Let    $\theta \in \Sigma^*$ be such that   $|\theta|=|\mathcal A|^2+1$ and $\alpha_i=\alpha'_i\theta$,  with $\alpha'_i \prec \alpha'_{i+1}\prec \ldots$.   Since $\alpha'_i\equiv_{\mathcal L}\alpha'_j$ implies  $\alpha_i\equiv_{\mathcal L}\alpha_j$,  the monotone sequence $(\alpha'_i)_{i\geq 1}$  is also not eventually constant modulo $\equiv_{\mathcal L}$. 
Since the set of $ \mathcal A $'s states is finite,  and  $(\alpha'_i)_{i\geq 1}$  is  not eventually constant modulo $\equiv_{\mathcal L}$,  by possibly considering a subsequence of  $(\alpha'_i)_{i\geq 1}$ we can further suppose that  all    elements   of  odd index end  in the same state $x'$,  all   elements  of  even index end the same state $y'$,  and  $x'\neq y'$.

Let $m = |\mathcal A|^2$, and consider  the  last $|\mathcal A|^2+1$ states $x'=x_0, x_{1}, \ldots, x_{m}$ of the $ \alpha_{1} $-labelled path  from the initial state  $s$. Note that all $\alpha_i$'s with odd $i$ share this path. 
Similarly, consider   the last $|\mathcal A|^2+1$  states $y'=y_{0}, y_{1}, \ldots, y_{m}$ of the $ \alpha_{2} $-labelled path  from the initial state  $s$. 
Again, all  $\alpha_i$'s  with  even $i$ share this path. Moreover, both  paths are labelled by the same word $\theta$ and  $x_k\neq y_k$, for all $k=0, \ldots m$ (otherwise  the sequence $(\alpha_i)_{i\geq 1}$ would be eventually constant, which is not).

Since $|\theta|=m+1=|\mathcal A|^2+1$,    we can find $i_0,n_0$ with     $0\leq i_0<n_0\leq m$ such that 
$(x_{i_0}, y_{i_0})= (x_{n_0}, y_{n_0})$,  that is,   the two subpaths    \[ x_{i_{0}}, x_{i_{0}+1}, \ldots , x_{n_0}=x_{i_0}, \]
\[y_{i_{0}}, y_{i_{0+1}}, \ldots ,y_{n_0}=y_{i_0},\]   are cycles of the same length  labelled by the same word, say $\gamma'$.     Note that $|\gamma'|\leq |\mathcal A|^2$.

Since $ \gamma' $ is a factor of $ \theta $,  there exist $\eta, \delta \in \Sigma^*$ such that  $\theta= \eta \gamma' \delta$. All $ \alpha'_{i}\eta$'s with $i$ odd  end in  $x_{i_0}$ and all  $ \alpha'_{i}\eta$'s  with $i$ even end  in   $y_{i_0}$, with $x_{i_{0}}\neq y_{i_0}$. Moreover, $\gamma'$ labels  two cycles starting in $x_{i_0}$ and  $y_{i_0}$, respectively. 

\medskip

Let  $\alpha= \alpha'_{1}\eta$, $\beta= \alpha'_{2}\eta$, $\alpha'= \alpha'_{3}\eta$, and note that $\alpha, \beta, \alpha'$ satisfies the first three   properties of our Claim,   with
$u=x_{i_0}$ and $v=y_{i_0}$. 

\medskip
 
 We now  prove the last point of Claim \ref{firstclaim}, that is,    we can also limit, effectively,  the lengths of $\alpha, \beta$, and $ \alpha'$. 
 Given a word $\varphi \in \Sigma^*$ and $k\geq 1$  we denote by $\varphi(k)$ the $k$-th letter from the right, whenever $|\varphi|\leq k$, or 
 the empty word   $\epsilon$, otherwise  (e.g. $\varphi(1)$ is the last letter of $\varphi$). 

Given $\alpha, \beta, \alpha',\gamma,u,v$  as defined above, let  $d_{\alpha',\beta}$ be the first position from the right in which $\alpha'$ and $\beta$ differ.  Since $\alpha'\succ \beta$, we have $|\alpha'|\geq d_{\alpha',\beta}$.
Similarly,  let $d_{\beta,\alpha}$ be the first position from the right in which $\beta$ and $\alpha$ differ. Again, since $\beta \succ \alpha$, we have $|\beta |\geq  d_{\beta, \alpha}$.
Proceeding by cases, consider:  
\begin{enumerate}
\item [i)] $d_{\alpha',\beta}\leq d_{\beta,\alpha}$ (so that the position in $ \beta $ from the right are: $ \ldots \ldots   d_{\beta,\alpha}  \ldots \ldots   d_{\alpha',\beta}  \ldots 2,  1 $);
\item [ii)] $ d_{\beta,\alpha}<d_{\alpha',\beta}$ (so that the position in $ \beta $ from the right are: $ \ldots    d_{\alpha',\beta}  \ldots   d_{\beta,\alpha} \ldots \ldots \ldots 2,  1 $).
  \end{enumerate}
 
In case i),  since     $|\alpha'|\geq  d_{\alpha',\beta}$,  $|\beta| \geq  d_{\beta,\alpha}\geq d_{\alpha',\beta}$,    the words  $\alpha', \beta$, and $ \alpha$ end with  the same word $\xi$ with $|\xi|= d_{\alpha',\beta}-1$. Since  
$\beta\prec \alpha'$ it must be that 
$ \beta(d_{\alpha',\beta})\prec  \alpha'(d_{\alpha',\beta})$. That is, for some $ \phi, \phi', \psi \in \Sigma^{*} $:
\begin{equation}\label{eq}
\alpha=\phi \alpha(d_{\alpha',\beta}) \xi\prec \beta=\psi  \beta(d_{\alpha',\beta}) \xi \prec  \alpha'=\phi' \alpha'(d_{\alpha',\beta})\xi ,  
\end{equation}
with $ \phi \alpha(d_{\alpha',\beta}) =\epsilon $, whenever $ |\alpha|< d_{\alpha',\beta} .$ 

 See Figure \ref{figura-i}.

\begin{figure}[H]
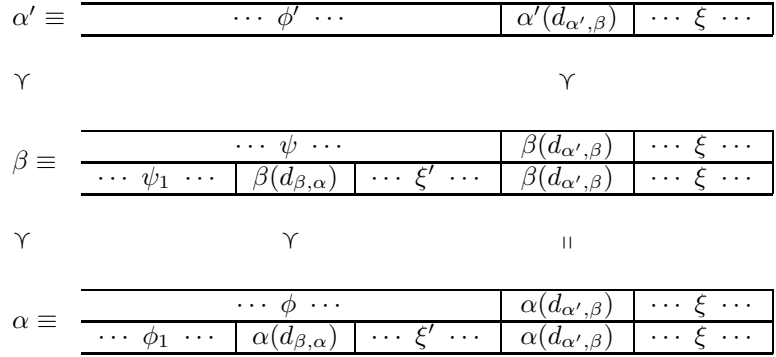

\centering
\begin{tabular}{lc|c|c|c|c|}
\cline{2-6}
$ \alpha' \equiv $ & \multicolumn{3}{c|}{$ \cdots \  \phi' \ \cdots $}  & $ \alpha'(d_{\alpha',\beta}) $ & $ \cdots \  \xi \  \cdots $ \\ \cline{2-6}
\multicolumn{6}{c}{}\\ 
$ \curlyvee $ &\multicolumn{3}{c}{ }& \multicolumn{1}{c}{$ \curlyvee $} & \multicolumn{1}{c}{}\\ 
\multicolumn{6}{c}{}\\ \cline{2-6}
\multirow{2}{*}{$ \beta \equiv $} & \multicolumn{3}{c|}{$\cdots \  \psi \ \cdots  $} & $ \beta(d_{\alpha',\beta}) $ & $ \cdots \  \xi \  \cdots $ \\ \cline{2-6}
 & $ \cdots \  \psi_{1}\ \cdots $ & $ \beta(d_{\beta,\alpha}) $ & $ \cdots \ \xi' \  \cdots $ & $ \beta(d_{\alpha',\beta}) $ & $ \cdots \  \xi \  \cdots $ \\ \cline{2-6}
\multicolumn{6}{c}{}\\ 
$ \curlyvee $ &\multicolumn{3}{c}{$ \curlyvee $ }& \multicolumn{1}{c}{$ \shortparallel $} & \multicolumn{1}{c}{}\\ 
\multicolumn{6}{c}{}\\ \cline{2-6}
\multirow{2}{*}{$ \alpha \equiv $} & \multicolumn{3}{c|}{$\cdots \  \phi \ \cdots  $} & $ \alpha(d_{\alpha',\beta}) $ & $ \cdots \  \xi \  \cdots $ \\ \cline{2-6}
& $ \cdots \  \phi_{1}\ \cdots $ & $ \alpha(d_{\beta,\alpha}) $ & $ \cdots \ \xi' \  \cdots $ & $ \alpha(d_{\alpha',\beta}) $ & $ \cdots \  \xi \  \cdots $ \\ \cline{2-6}
\end{tabular}
\caption{Case i) with $ d_{\alpha',\beta} < d_{\beta,\alpha} $. }\label{figura-i}
\end{figure}
 
We can assume, without loss of generality,  that   $|\xi|  \leq |\mathcal A|^3$. In fact, 
 if $|\xi|>|\mathcal A|^3$  then,  considering the triples of states visited simultaneously while reading the last  $d_{\alpha',\beta}$'s letters of  $\alpha, \beta, \alpha'$, respectively,  we should meet  a repetition. If this were the case,
 we could erase a common factor from $\xi$, obtaining a shorter word $\xi_1$ such that 
 $\phi \alpha(d_{\alpha',\beta})\xi_1\prec \psi \beta(d_{\alpha',\beta}) \xi_1\prec \phi'\alpha'(d_{\alpha',\beta}) \xi_1$, with the three paths  still ending in $u,v$, and $u$, respectively.   Hence, we may suppose  $|\xi|\leq |\mathcal A|^3$ in  (\ref{eq}) above.  
  
  Consider now the case  $d_{\alpha',\beta}=d_{\beta,\alpha}$. Let  $s_1,s_2$, and $s_3$  be the states reached from $ s $ by reading $ \phi, \psi$, and $ \phi'$, respectively. Since   $d_{\beta,\alpha}=d_{\alpha',\beta}$ is a position on the right of $ \phi, \psi$ and $\phi' $ in $ \alpha, \beta $, and $ \alpha' $, respectively,   we may suppose w.l.o.g.   that $ \phi, \psi_1, \phi_1$   label simple paths leading from  $s$ to  $s_1,s_2,s_3$, so that $| \phi,| |\psi_1|, |\phi_1|\leq |\mathcal A|$. Hence in this case we have $|\alpha|, |\beta|, |\alpha'|\leq  1+|\mathcal A| +|\mathcal A|^3$. 
 
Next,  consider the case  $d_{\alpha',\beta}< d_{\beta,\alpha}$. In this case,   $\phi, \psi$ end  with  the same word $\xi'$   with $|\xi'|= d_{\beta,\alpha}-d_{\alpha',\beta}-1$, and  
\[  \phi=\phi_1 \alpha(d_{\beta,\alpha})\xi' \prec  \psi=\psi_1\beta(d_{\beta,\alpha})\xi' \] 
(see the picture above).
We may assume, without loss of generality, that   $|\xi'| \leq |\mathcal A|^2 $. In fact, 
 if $|\xi|>|\mathcal A|^2$  then, reasoning as above but considering pairs of states instead of triples,  
  we could  erase a common factor from $\xi'$, obtaining a shorter word $\xi'_1$ such that 
\begin{equation}
\label{eq2}
\phi_1  \alpha(d_{\beta,\alpha})\xi'_1 \alpha(d_{\alpha',\beta})\xi\prec \psi_1 \beta(d_{\beta,\alpha})\xi'_1 \beta(d_{\alpha',\beta}) \xi \prec  \alpha'=\phi' \alpha'(d_{\alpha',\beta})\xi,\end{equation}
where  the words above   still end  in $u,v$, and $u$, respectively. By repeating the same argument, we see that we may suppose  $|\xi_1|\leq |\mathcal A|^2$  and 
$|\xi|\leq |\mathcal A|^3$ in (\ref{eq2}).

Consider now the states $s_1,s_2$, and $s_3$ reached by reading $ \phi_{1}, \psi_{1} $, and $ \phi' $ from $ s $, respectively. Since  $ d_{\beta,\alpha} $ is a position on the right of $ \phi_{1} $ and $ \psi_{1} $, and $ d_{\alpha',\beta} $ is a position on the right of $ \psi_{1} $ and $ \phi' $, we may assume, without loss of generality, that  $ \phi_1, \psi_1$, and $ \phi'$   label simple paths leading from  $s$ to  $s_1,s_2$, and $s_3$, respectively. Hence $| \phi_{1}|,  |\psi_1|, |\phi'|\leq |\mathcal A|$.
 
Summarising, we may suppose $|\alpha|, |\beta|, |\alpha'|\leq 2+ |\mathcal A|+|\mathcal A|^2+|\mathcal A|^3$, ending the proof of case i).

 \medskip

Case ii), in which $d_{\alpha',\beta}> d_{\beta,\alpha}$, can be treated analogously. The skeptical reader can consult the following graphic proof (see Figure \ref{figura-ii}) and    this ends the proof of Claim \ref{firstclaim}. See Figure \ref{figura-ii}

\begin{figure}[H]
\centering
\begin{tabular}{lc|c|c|c|c|}
\cline{2-6}
\multirow{2}{*}{$ \alpha' \equiv $} & \multicolumn{3}{c|}{$\cdots \  \phi' \ \cdots  $} & $ \alpha'(d_{\beta,\alpha}) $ & $ \cdots \  \xi \  \cdots $ \\ \cline{2-6}
& $ \cdots \  \phi_{1}'\ \cdots $ & $  \alpha'(d_{\alpha',\beta}) $ & $ \cdots \ \xi' \  \cdots $ & $ \alpha'(d_{\beta,\alpha}) $ & $ \cdots \  \xi \  \cdots $ \\ \cline{2-6}
\multicolumn{6}{c}{}\\ 
$ \curlyvee $ &\multicolumn{3}{c}{$ \curlyvee $ }& \multicolumn{1}{c}{$ \shortparallel $} & \multicolumn{1}{c}{}\\  
\multicolumn{6}{c}{}\\ \cline{2-6}
\multirow{2}{*}{$ \beta \equiv $} & \multicolumn{3}{c|}{$\cdots \  \psi \ \cdots  $} & $ \beta(d_{\beta,\alpha}) $ & $ \cdots \  \xi \  \cdots $ \\ \cline{2-6}
 & $ \cdots \  \psi_{1}\ \cdots $ & $ \beta(d_{\alpha',\beta}) $ & $ \cdots \ \xi' \  \cdots $ & $ \beta(d_{\beta,\alpha}) $ & $ \cdots \  \xi \  \cdots $ \\ \cline{2-6}
\multicolumn{6}{c}{}\\ 
$ \curlyvee $ &\multicolumn{3}{c}{ }& \multicolumn{1}{c}{$ \curlyvee $} & \multicolumn{1}{c}{}\\
\multicolumn{6}{c}{}\\ \cline{2-6}
$ \alpha \equiv $ & \multicolumn{3}{c|}{$ \cdots \  \phi \ \cdots $}  & $ \alpha(d_{\beta,\alpha}) $ & $ \cdots \  \xi \  \cdots $ \\ \cline{2-6}
\multicolumn{6}{c}{}\\ 
\end{tabular}
\caption{Case ii). }\label{figura-ii}
\end{figure}

Turning now to our main claim, if  $\alpha, \beta, \alpha'$, and $ \gamma' $ are as in Claim \ref{firstclaim},   let $h$ be the minimum number  such that $|\alpha|, |\beta|, |\alpha'|< h|\gamma'|$.  If $\gamma=(\gamma')^h$, then $|\gamma |=h|\gamma'|$ and   $|\alpha|, |\beta|, |\alpha'|<  |\gamma|$.
 
 
  We consider two cases:
 \begin{enumerate}
 \item $\gamma\prec \beta$; in this case   we define 
 $\mu=\beta$ and $ \nu=\alpha'$,  so  that $\gamma \prec \mu\prec  \nu $;
 \item $\beta\prec \gamma$;   in this case   we define  $\mu=\alpha$ and $ \nu=\beta $, so that $  \mu\prec  \nu \prec \gamma$.
 \end{enumerate}
 
 From  $|\alpha|, |\beta|, |\alpha'|<  2+ |\mathcal A|+|\mathcal A|^2+|\mathcal A|^3$ it follows that
 \[|\gamma|\leq max\{|\alpha|, |\beta|, |\alpha'|\}+ |\gamma'|\leq
 (2+ |\mathcal A|+|\mathcal A|^2+|\mathcal A|^3) +|\mathcal A|^2=
 2+ |\mathcal A|+2|\mathcal A|^2+|\mathcal A|^3\]
 and hence $\mu,\nu,\gamma$ satisfies the required properties. 

\end{proof}

We now use the preceding theorem to prove the decidability of being a Wheeler language.

 \begin{theorem}\label{thm:L Wheeler polytime}
 We can decide whether the regular language $\mathcal L$ accepted by a given edge-labeled DFA $\mathcal A$ is Wheeler in polynomial time.
 \end{theorem}
  \begin{proof}  Since the construction of the minimum automaton recognizing a language   can  be done in polynomial time starting from a DFA recognizing it,  we may suppose that $\mathcal A$ is minimum. 
 We exhibit a dynamic programming algorithm that finds $\mu$, $\nu$, and $\gamma$ satisfying Theorem \ref{decidability} if and only if such strings exist. Let $N =  2+|\mathcal A|+2|\mathcal A|^2+|\mathcal A|^3$ be the (polynomial) upper bound to the length of those strings. 
 
 We consider only the case $\mu, \nu \prec \gamma$, as the other can be solved symmetrically. 
 Let $\pi_{u,\ell}$, with $u\in Q - \{s\}$  and $2 \leq \ell \leq N$, denote the predecessor of $u$ such that the co-lexicographically smallest path of length (number of nodes) $\ell$ connecting the source $s$ to $u$ passes through $\pi_{u,\ell}$ as follows: $s \rightsquigarrow \pi_{u,\ell} \rightarrow u$. 
 The node $\pi_{u,\ell}$ coincides with $s$ if $\ell=2$ and $u$ is a successor of $s$; in this case, the path is simply $s \rightarrow u$.
 If there is no path of length $\ell$ connecting $s$ with $u$, then $\pi_{u,\ell} = \bot$.
 Note that the set $\{\pi_{u,\ell}\ :\ 2\leq \ell \leq N,\ u\in Q - \{s\}\}$ stores in just polynomial space all co-lexicographically smallest paths of any fixed length $\ell \leq N$ from the source to any node $u$. We denote such path with $\alpha_\ell(u)$, and the corresponding sequence of labels with $\lambda(\alpha_\ell(u))$ (that is, the sequence of $\ell-1$ symbols labeling the path's edges). 
 Note that $\alpha_\ell(u)$ can be obtained recursively (in $O(\ell)$ steps) as $\alpha_\ell(u) = \alpha_{\ell-1}(\pi_{u,\ell}) \rightarrow u$, where $\alpha_{1}(s) = s$ by convention.
 
 Clearly, each $\pi_{u,\ell}$ can be computed in polynomial time using dynamic programming. First, we set $\pi_{u,2} = s$ for all successors $u$ of $s$. Then, for $\ell = 3, \dots, N$:
 
$$
\pi_{u,\ell} = \underset{v\in Pred(u)}{\mathrm{argmin}}
  \lambda( \alpha_{\ell-1}(v) ) \cdot \lambda(v,u)    
$$

 where $Pred(u)$ is the set of all predecessors of $u$ and the $\mathrm{argmin}$ operator compares strings in co-lexicographic order. In the equation above, if none of the $\alpha_{\ell-1}(v)$ are well-defined (because there is no path of length $\ell-1$ from $s$ to $v$), then $\pi_{u,\ell} = \bot$.

 The second (similar) ingredient is to compute pairs $\psi_{u,u',v,v',\ell} = \langle u'', v''\rangle$, with $u,u',v,v'\in Q-\{s\}$
 and $2\leq \ell \leq N$, such that:
 
 \begin{enumerate}
     \item $u''$ is a predecessor of $u'$ and $v''$ is a predecessor of $v'$,
     \item $\lambda(u'',u') = \lambda(v'',v') = c$, for some $c\in \Sigma$,
     \item there exist two paths of length (number of nodes) $\ell-1$ from $u$ to $u''$ and from $v$ to $v''$ labelled with the same string $\beta$ (if $\ell=2$, then $\beta=\epsilon$), and
     \item $\langle u'', v''\rangle$ is chosen so that $\beta\cdot c$ is co-lexicographically maximum.
 \end{enumerate}
 
 As before, if such two paths and such a $\beta$ do not exist, then $\psi_{u,u',v,v',\ell} = \bot$. Moreover, if $(u,u')$ and $(v,v')$ are edges 
 with $\lambda(u,u') = \lambda(v,v')$,
 then $\psi_{u,u',v,v',2} = \langle u, v\rangle$ and the two associated paths are $u \rightarrow u'$ and $v \rightarrow v'$. 
 
 Analogously to the (simpler) case seen before, these pairs store in polynomial space, for each  $u,u',v,v'$ and length $\ell$, the co-lexicographically largest string of length $\ell-1$ labeling two paths $u \rightsquigarrow u'$ and $v \rightsquigarrow v'$, as well as the two paths themselves. We denote these two paths as $\beta_\ell(\underline u,\underline u',v,v')$ and $\beta_\ell(u, u',\underline v,\underline v')$, respectively. Note that, by our definition, $\lambda(\beta_\ell(\underline u,\underline u',v,v')) = \lambda(\beta_\ell(u, u',\underline v,\underline v'))$. Again, these paths can be obtained in a recursive fashion using the pairs.
 
 Pairs $\psi_{u,u',v,v',\ell} = \langle u'', v''\rangle$ can be computed in polynomial time using dynamic programming as follows.
 We set all $\psi_{u,u',v,v',2} = \langle u,v\rangle$ whenever  $(u,u')$ and $(v,v')$ are edges with $\lambda(u,u') = \lambda(v,v')$ ($\bot$ otherwise) and, for $\ell = 3, \dots, N$:
 
$$
    \psi_{u,u',v,v',\ell} = \underset{\langle u'',v''\rangle \in Pred(u')\times Pred(v')\ :\ \lambda(u'',u')=\lambda(v'',v') }{\mathrm{argmax}} \lambda(\beta_{\ell-1}(\underline u,\underline u'',v,v'')) \cdot \lambda(u'',u')
$$

 where the $\mathrm{argmax}$ operator compares strings in co-lexicographic order.

To conclude, in order to check the conditions of Theorem \ref{decidability}, we proceed as follows. First, we guess the nodes $u$ and $v$ and the lengths $|\mu|, |\nu|<|\gamma|\leq 2(2+|\mathcal A|+|\mathcal A|^2+|\mathcal A|^3)$ (there are only polynomially-many choices to try).
Then:
\begin{enumerate}
    \item We compute the co-lexicographically smallest $\mu' = \lambda(\alpha_{|\mu|}(u))$ labeling a path of length $|\mu|$ from $s$ to $u$,
    \item we compute the co-lexicographically smallest $\nu' = \lambda(\alpha_{|\nu|}(v))$ labeling a path of length $|\nu|$ from $s$ to $v$,
    \item we compute the co-lexicographically largest $\gamma' = \lambda(\beta_{|\gamma|}(\underline u,\underline u,v,v))$  labeling two paths of length $|\gamma|$ from $u$ to $u$ and from $v$ to $v$ (that is, two cycles), and
    \item we check if $\mu',\nu' \prec \gamma'$. We declare $\mathcal L(\mathcal A)$ non Wheeler if and only if this test succeeds for at least one choice of $u,v,|\mu|, |\nu|, |\gamma|$. 
\end{enumerate}

Clearly, the existence of $\mu, \nu$, and $\gamma$ implies that $\mu', \nu'$, and $\gamma'$ exist and that they satisfy the conditions of Theorem \ref{decidability}: we have $\mu' \preceq \mu$, $\nu' \preceq \nu$, $\gamma' \succeq \gamma$, and $\mu,\nu \prec \gamma$, therefore $\mu',\nu' \prec \gamma'$ holds. Conversely, the theorem states that if we find such $\mu', \nu'$, and $\gamma'$ then the original language is not Wheeler. 

\end{proof} 
 
In~\cite{alanko2020regular} it is presented a procedure for obtaining the minimum WDFA equivalent to a given acyclic DFA. We now show that, while a more general procedure for converting any DFA recognizing a Wheeler language into the minimum equivalent WDFA would solve the problem of Theorem \ref{thm:L Wheeler polytime}, it would take exponential time in the worst case (as opposed to Theorem \ref{thm:L Wheeler polytime}) just to produce the output WDFA (or to decide that such a WDFA does not exist): there exists a family of regular languages where the size of the smallest WDFA is exponential in the size of the smallest equivalent DFA. Consider the family of languages $L_1,L_2,\ldots$, where $L_m = \{ c \alpha e \; | \; \alpha \in \{a,b\}^m \} \cup \{ d \alpha f \; | \; \alpha \in \{a,b\}^m \}$.
Figure \ref{fig:dfa_wdfa_worst_case2} shows a DFA and the smallest WDFA for the language $L_3$. In general, we can build a DFA for $L_m$ by generalizing the construction in the figure: the source node has outgoing edges labeled with $c$ and $d$, followed by simple linear size "universal gadgets" capable of generating all binary strings of length $m$, with one gadget followed by an $e$ and the other by an $f$. The two sink states are the only accepting states.

The smallest WDFA for $L_m$ is an unraveling of the described DFA, such that all paths up to (but not including) the sinks end up in distinct nodes, i.e. the universal gadgets are replaced by full binary trees (see Figure \ref{fig:dfa_wdfa_worst_case2}). It is easy to see that the automaton is Wheeler as the only nodes that have multiple incoming paths are the sinks, and the sinks have unique labels.

By \cite[Thm. 4.2]{alanko2020regular}, to prove that this is the minimum WDFA we need to check that all colexicographically consecutive pairs of nodes with the same incoming label are Myhill-Nerode inequivalent. As labels $c,d,e$ and $f$ occur only once, it is enough to focus on nodes that have label $a$ or $b$. Let $B_1, B_2, B_{2^{m+1}-1}$ be the colexicographically sorted sequence of all possible binary strings with lengths $1 \leq |B_i| \leq m$ from the alphabet $\{a,b\}$. Observe that the nodes with incoming label $a$ and $b$ correspond to path labels of the form $c B_i$ and $dB_i$ for all $1 \leq i \leq 2^{m+1}-1$. The co-lexicographically sorted order of these path labels is:
$$c B_1 < d B_1 < c B_2 < d B_2 < \ldots < c B_{2^{m+1}-1} < d B_{2^{m+1}-1} $$
Here we can see that all consecutive pairs have a different first character: they therefore lead to a different sink in the construction and hence are not Myhill-Nerode equivalent. We therefore conclude that the automaton is the minimum WDFA.
The DFA has $n = 4m + 5$ states and the WDFA has $1 + 2^{m+2} = 1 + 2^{(n-5)/4 + 2}$ states, so we obtain the following result:

\begin{theorem}
	The minimum WDFA equivalent to a DFA with $n$ states has $\Omega(2^{n/4})$ states in the worst case.
\end{theorem}

\begin{figure}
	
	\centering
	
	\begin{tikzpicture}[->,>=stealth', semithick, auto, scale=0.55]
	\tikzstyle{every state}=[scale=0.4]
	
	
	\node[state, accepting, label=above:{}] (T21)    at (-4,0+2)		    {\Huge $e$};
	
	
	\node[state, label=above:{}] (T18) at (-5,-1+2) 		{\huge $a$};
	\node[state, label=above:{}] (T17) at (-5,1+2)   		{\huge $b$};  
	
	\node[state, label=above:{}] (T16) at (-6,-1+2) 		{\huge $a$};
	\node[state, label=above:{}] (T15) at (-6,1+2)   		{\huge $b$};  
	
	\node[state, label=above:{}] (T14) at (-7,-1+2) 		{\huge $a$};
	\node[state, label=above:{}] (T13) at (-7,1+2)   		{\huge $b$};  
	
	\node[state, label=above:{}] (T12) at (-8,-0+2)   		{\huge $c$};

	
	\node[state, accepting, label=above:{}] (T10)    at (-4,-4+2)	    {\huge $f$};
	
	
	\node[state, label=above:{}] (T7) at (-5,-5+2) 		{\huge $a$};
	\node[state, label=above:{}] (T6) at (-5,-3+2)   		{\huge $b$};  
	
	\node[state, label=above:{}] (T5) at (-6,-5+2) 		{\huge $a$};
	\node[state, label=above:{}] (T4) at (-6,-3+2)   		{\huge $b$};  
	
	\node[state, label=above:{}] (T3) at (-7,-5+2)    		{\huge $a$};
	\node[state, label=above:{}] (T2) at (-7,-3+2)   		{\huge $b$};  
	
	\node[state, label=above:{}] (T1) at (-8,-4+2)   		{\huge $d$};

	\node[state, label=above:{}] (T0) at (-9,-2+2)   		{};  
	
	\draw (T0) edge node {} (T1);
	
	\draw (T1) edge node {} (T2);
	\draw (T1) edge node {} (T3);
	
	\draw (T2) edge node {} (T4);
	\draw (T2) edge node {} (T5);
	\draw (T3) edge node {} (T4);
	\draw (T3) edge node {} (T5); 
	
	\draw (T4) edge node {} (T6);
	\draw (T4) edge node {} (T7);
	\draw (T5) edge node {} (T6);
	\draw (T5) edge node {} (T7);  
	
	\draw (T6) edge node {} (T10);
	\draw (T6) edge node {} (T10);
	\draw (T7) edge node {} (T10);
	\draw (T7) edge node {} (T10);  
	
	

	\draw (T0) edge node {} (T12);
	
	\draw (T12) edge node {} (T13);
	\draw (T12) edge node {} (T14);
	
	\draw (T13) edge node {} (T15);
	\draw (T13) edge node {} (T16);
	\draw (T14) edge node {} (T15);
	\draw (T14) edge node {} (T16); 
	
	\draw (T15) edge node {} (T17);
	\draw (T15) edge node {} (T18);
	\draw (T16) edge node {} (T17);
	\draw (T16) edge node {} (T18);  
	
	\draw (T17) edge node {} (T21);
	\draw (T17) edge node {} (T21);
	\draw (T18) edge node {} (T21);
	\draw (T18) edge node {} (T21);  
	
	
	
	\node[state, label=above:{}] (S0) at (1-3,0)   		{};  
	
	\node[state, label=above:{}] (S11) at (2.5-3,1/2)   		{\huge $c$};
	\node[state, label=above:{}] (S21) at (4-3,1/2)   		{\huge $a$};
	\node[state, label=above:{}] (S22) at (4-3,5/2)   		{\huge $b$};
	\node[state, label=above:{}] (S31) at (5.5-3,1/2)   		{\huge $a$};
	\node[state, label=above:{}] (S32) at (5.5-3,3/2)   		{\huge $b$};
	\node[state, label=above:{}] (S33) at (5.5-3,5/2)   		{\huge $a$};
	\node[state, label=above:{}] (S34) at (5.5-3,7/2)   		{\huge $b$};
	\node[state, label=above:{}] (S41) at (7.0-3,1/2)   		{\huge $a$};
	\node[state, label=above:{}] (S42) at (7.0-3,2/2+0.2)   		{\huge $b$}; 
	\node[state, label=above:{}] (S43) at (7.0-3,3/2+0.4)   		{\huge $a$}; 
	\node[state, label=above:{}] (S44) at (7.0-3,4/2+0.6)   		{\huge $b$}; 
	\node[state, label=above:{}] (S45) at (7.0-3,5/2+0.8)   		{\huge $a$}; 
	\node[state, label=above:{}] (S46) at (7.0-3,6/2+1.0)   		{\huge $b$}; 
	\node[state, label=above:{}] (S47) at (7.0-3,7/2+1.2)   		{\huge $a$}; 
	\node[state, label=above:{}] (S48) at (7.0-3,8/2+1.4)   		{\huge $b$};
	\node[state, accepting, label=above:{}] (Se) at (8.5-3,4.5/2)   		{\Huge $e$};
	
	\node[state, label=above:{}] (D11) at (2.5-3,-1/2)   		{\huge $d$};
	\node[state, label=above:{}] (D21) at (4-3,-1/2)   		{\huge $a$};
	\node[state, label=above:{}] (D22) at (4-3,-5/2)   		{\huge $b$};
	\node[state, label=above:{}] (D31) at (5.5-3,-1/2)   		{\huge $a$};
	\node[state, label=above:{}] (D32) at (5.5-3,-3/2)   		{\huge $b$};
	\node[state, label=above:{}] (D33) at (5.5-3,-5/2)   		{\huge $a$};
	\node[state, label=above:{}] (D34) at (5.5-3,-7/2)   		{\huge $b$};
	\node[state, label=above:{}] (D41) at (7.0-3,-1/2)   		{\huge $a$};
	\node[state, label=above:{}] (D42) at (7.0-3,-2/2-0.2)   		{\huge $b$}; 
	\node[state, label=above:{}] (D43) at (7.0-3,-3/2-0.4)   		{\huge $a$}; 
	\node[state, label=above:{}] (D44) at (7.0-3,-4/2-0.6)   		{\huge $b$}; 
	\node[state, label=above:{}] (D45) at (7.0-3,-5/2-0.8)   		{\huge $a$}; 
	\node[state, label=above:{}] (D46) at (7.0-3,-6/2-1.0)   		{\huge $b$}; 
	\node[state, label=above:{}] (D47) at (7.0-3,-7/2-1.2)   		{\huge $a$}; 
	\node[state, label=above:{}] (D48) at (7.0-3,-8/2-1.4)   		{\huge $b$};
	\node[state, accepting, label=above:{}] (Df) at (8.5-3,-4.5/2)   		{\huge $f$};
	
	\draw (S0) edge node {} (S11);
	\draw (S11) edge node {} (S21);
	\draw (S11) edge node {} (S22);
	\draw (S21) edge node {} (S31);
	\draw (S21) edge node {} (S32);
	\draw (S22) edge node {} (S33);
	\draw (S22) edge node {} (S34);
	\draw (S31) edge node {} (S41);
	\draw (S31) edge node {} (S42);
	\draw (S32) edge node {} (S43);
	\draw (S32) edge node {} (S44);
	\draw (S33) edge node {} (S45);
	\draw (S33) edge node {} (S46);
	\draw (S34) edge node {} (S47);
	\draw (S34) edge node {} (S48);
	\draw (S41) edge node {} (Se);
	\draw (S42) edge node {} (Se);
	\draw (S43) edge node {} (Se);
	\draw (S44) edge node {} (Se);
	\draw (S45) edge node {} (Se);
	\draw (S46) edge node {} (Se);
	\draw (S47) edge node {} (Se);
	\draw (S48) edge node {} (Se);
	
	\draw (S0) edge node {} (D11);
	\draw (D11) edge node {} (D21);
	\draw (D11) edge node {} (D22);
	\draw (D21) edge node {} (D31);
	\draw (D21) edge node {} (D32);
	\draw (D22) edge node {} (D33);
	\draw (D22) edge node {} (D34);
	\draw (D31) edge node {} (D41);
	\draw (D31) edge node {} (D42);
	\draw (D32) edge node {} (D43);
	\draw (D32) edge node {} (D44);
	\draw (D33) edge node {} (D45);
	\draw (D33) edge node {} (D46);
	\draw (D34) edge node {} (D47);
	\draw (D34) edge node {} (D48);
	\draw (D41) edge node {} (Df);
	\draw (D42) edge node {} (Df);
	\draw (D43) edge node {} (Df);
	\draw (D44) edge node {} (Df);
	\draw (D45) edge node {} (Df);
	\draw (D46) edge node {} (Df);
	\draw (D47) edge node {} (Df);
	\draw (D48) edge node {} (Df);

	\end{tikzpicture}
	
	\caption{Left: a DFA recognizing $L_3$. Right: the minimum WDFA recognizing $L_3$. For clarity the labels are drawn on the nodes: the label of an edge is the label of the destination node.}
	\label{fig:dfa_wdfa_worst_case2}
	
\end{figure}
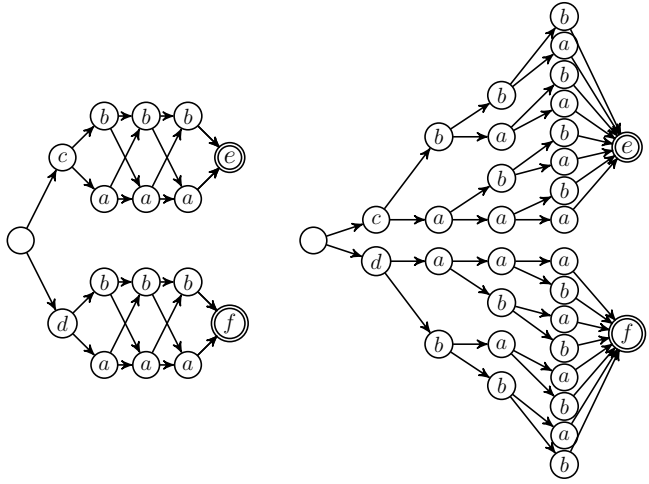

\subsection{Is a $ \mathcal A $  Wheeler?} \label{AWheeler}

In this section  we consider   the problem of deciding whether a given NFA can be endowed with  a Wheeler order. In this case, since the problem is obviously decidable, we 
are interested in  its  complexity. 
Since input-consistency is a necessary condition for Wheelerness, without loss of generality in this section we will assume that the input NFA is state-labeled.

The problem has already been considered in \cite{alanko2020regular, DBLP:conf/esa/GibneyT19}, where the following results can be found: let d-NFA denote the class of NFA's with at most d equally-labelled transitions leaving any state.
\begin{enumerate}
\item (\cite{alanko2020regular}) The problem of recognizing and sorting Wheeler d-NFA's is in $P$  for $d \leq 2$ (in particular,  it is in $P$ for deterministic automata, which correspond to the class of  1-NFA). 
\item (\cite{DBLP:conf/esa/GibneyT19}) shows that the problem is NP-complete for $d \geq 5$.
\end{enumerate}

Here we see that   NP completeness depends  on redundancies of NFA: in fact, we shall prove that the problem of deciding whether a given  reduced NFA   (see Def. \ref{red}) can be endowed with  a Wheeler order is in $P$. 
 
Let $\mathcal A=(Q,s, \delta,F)$, with $|Q|=n$,  be an input-consistent 
NFA automaton (with no edges entering in the initial state $s$) over a finite ordered alphabet $\Sigma= \{a_1,  \ldots,a_k\}$, with $a_1\prec  \ldots\prec  a_k$.  Let $\lambda(u)$ be the label of (all) the edges entering  $u$,  $Q_a=\{u\in Q: \lambda(u)=a\}$,  $Q_\epsilon=\{s\}$; if $C\subseteq Q$ then let  $\delta_a[C]=\{q' \in Q: \exists q\in C ~q'\in  \delta(q,a)\}$.

\begin{definition}
We say that a partition $ \mathcal C = \{ C_{1}, \ldots, C_{n}\} $ of the set of the automaton states is $a$-\emph{forward-stable}, for $ a \in \Sigma$, if and only if for all $ C_{i} ,C_{j} \in \mathcal C$, either $ \delta_{a}[C_{i}] \supseteq C_{j}$ or $ \delta_{a}[C_{i}] \cap C_{j}=\emptyset$.
\end{definition}

$ \mathcal C $ is \emph{forward-stable} with respect to $ \delta $ if and only if is $a$-forward-stable  for all $ a \in \Sigma $.

Consider the   algorithm  \ref{forward} below,  the ``Forward Algorithm''.

%

 \begin{algorithm}[th!]\label{forward}
	\caption{\texttt{Forward Algorithm}}
	\label{alg:sort}
	
	\SetKwInOut{Input}{input}
	\SetKwInOut{Output}{output}
	\SetSideCommentLeft
	\LinesNumbered
	
	\Input{A state-labeled NFA $\mathcal A$ }
	\Output{The coarsest forward-stable partition of  $\mathcal A$'s states and (possibly) a Wheeler order of its states.}
	\BlankLine
	\BlankLine
	
	$\mathcal C \leftarrow \langle Q_\epsilon, Q_{a_1}, \ldots, Q_{a_k} \rangle$\;
	
	\Repeat{$\mathcal C=\mathcal C_{old}$}{
	Set $ \neg R(C)$, for all $C\in  \mathcal C$ \Comment*[r]{\small $ R(\cdot)  $ stands for ``reached'}  
	$\mathcal C_{old} \leftarrow \mathcal C$\; 
	$C \leftarrow first(\mathcal C)$\;
	
	    \While{$\mathcal C=\mathcal C_{old}$ and $C\neq null$}{
	        \For{$  C'\in \mathcal C$} { 
	            $e = \lambda(C')$\Comment*[r]{\small determine the (unique) $e= \lambda(u) $, for $ u\in C' $}\label{forward:Q_e}
	            \eIf{$R(C')$}{
	                $C_1' \leftarrow C'\setminus \delta_{e}(C)$\;\label{forward:delta_e} 
	                $C'_2 \leftarrow \delta_{e}(C)\cap C'$\; 
	                $R(C_1'); R(C_2')$\;
	                }{
	                $C_1' \leftarrow \delta_{e}(C)\cap C'$\;
	                $C'_2 \leftarrow C'\setminus \delta_{e}(C)$\;
	                $R(C_1'); \neg  R(C'_2)$\;
	                }
	        $Insert(C_1',C_2', C', \mathcal C)$\Comment*[r]{\small replace $ C' $ with $ C_{1}',C_{2}' $ (in order), ignoring empty sets} 
	        }
	        $C=next(C,\mathcal C)$\;\label{forward:next}
	    }
	}
\end{algorithm}

 \begin{lemma} The Forward Algorithm terminates in $O(|Q|^2\cdot |\delta|)$ steps. \end{lemma}
 \begin{proof} After every iteration of the \texttt{repeat} command, the resulting partition is a refinement of the previous one, and the algorithm stops when we obtain the same partition of the previous iteration. Since the original partition can be refined at most $|Q|$ times, we have at most $|Q|$ iteration of the \texttt{repeat} command. 
 
 The \texttt{while} loop runs for at most $|Q|$ times as well: by Line \ref{forward:next} and by the \texttt{while} condition, in the worst case we perform one iteration per element of $\mathcal C$. Being $\mathcal C$ a partition of $Q$, its cardinality is bounded by $|Q|$.
 
 For each iteration of the \texttt{while} loop, in line \ref{forward:delta_e} we compute the outgoing arcs labeled $e$ of $C$, for each $C\in \mathcal C$. Overall, this amortizes to $O(|\delta|)$ time per \texttt{while} iteration. Similarly, in the \texttt{for} loop we visit all the nodes in $C'$, for each $C'\in \mathcal C$. This amortizes to $O(|Q|)$ time per \texttt{while} iteration. 
 
 Overall, we obtain complexity $O(|Q|^2\cdot |\delta|)$.
 
 \end{proof}

 \begin{lemma} \label{equireach} If $\mathcal C_{out}$ is the output  of the Forward Algorithm and $u, v \in C\in \mathcal C_{out}$, then 
\[\{\alpha: u\in  {\delta}(s, \alpha)\}=\{\alpha: v\in  {\delta}(s, \alpha)\}\]
 \end{lemma} 
 \begin{proof}Suppose, by way of a contradiction, that there exists a word $\alpha\in \Sigma^*$,  an element   $C\in \mathcal C_{out}$, and two states $u,v\in C$ such that $u\in{\delta}(s, \alpha), v\not \in {\delta}(s, \alpha)$. Consider  a word $ \alpha $ of minimal length having this property.  
 
Let $\alpha=\alpha'e$  and consider $u'\in  {\delta}(s, \alpha')$ such that  $u\in \delta(u',e)$. Let $C'\in  C_{out}$ be such that $u'\in C'$. Since $\mathcal C_{out}$ is the output of the algorithm, $C'$ cannot be a modifier for  $\mathcal C_{out}$; in particular, since $u \in \delta_e(C')\cap C\neq \emptyset$, we must have $C\subseteq \delta_e(C')$. Being $ v\in C $, there must exist $v'\in C'$ with $v\in \delta(v', e)$.  Since $ u', v' \in C' $ with $u'\in  {\delta}(s, \alpha')$,  by the minimality of $\alpha$ we have   $v'\in  {\delta}(s, \alpha')$.
 This implies $v\in {\delta}(s, \alpha)$, which contradicts our hypothesis.

 \end{proof}
 
 \begin{lemma} \label{agree} If $\mathcal A=(Q,s, \delta,<,F)$ is a Wheeler automaton, at any step of the Forward Algorithm the partition $\mathcal C=\langle C_1, \ldots, C_k\rangle$ agrees with the Wheeler order $<$ of the automaton, that is, if $i<j$, $u\in C_i, v\in C_j$ then $u<v$.
 \end{lemma}
 \begin{proof}  Reasoning by induction on the number of iterations of the repeat loop, observe that, by Wheeler (1), the initial partition  $\mathcal C=\langle Q_\epsilon, Q_{a_1}, \ldots, Q_{a_k} \rangle$  agrees with the Wheeler order $<$. 
 
Suppose 
$\mathcal C$  is the partition we obtain after an intermediate iteration. By induction, $\mathcal C$ agrees with the Wheeler order $<$. Let  $C=C_h\in \mathcal C$ be the   modifier chosen   and let $\mathcal C'$ be  the output of the repeat iteration using $ C_{h}$. We prove that $\mathcal C'$ still   agrees with  $<$. 
 Let $C'\in \mathcal C$ be such that  
 $\delta_e(C_h)\cap C'\neq \emptyset, ~~C'\setminus \delta_e(C_h)\neq \emptyset$ and consider the following two cases:
 \begin{itemize} 
 \item[$\neg R(C')$ ] Let  $x\in C_1'=\delta_e(C_h)\cap C', ~ y \in C'_2=C'\setminus \delta_e(C_h)$. We prove that 
$x<y$.  We begin observing that,   
 for all $k<h$,  we must have $\delta_e(C_k)\cap C'=\emptyset$. In fact, if this were not the case, we would have  had $C'\subseteq \delta_e(C_k)$ (or $C'$ would have been ``splitted'' in a previous step). But then, when $C_k$ was considered in the  while loop, at line 13 or 17 the algorithm would have set $R(C')$: a contradiction.  Hence,  $\delta_e(C_k)\cap C'=\emptyset$ for all $k<h$ and    any edge entering in  $y$ must start from an element $ y' \in C_{j} $ such that  $j>h$, that is: $y\in \delta(y',e)$. Since $ x\in C_1'=\delta_e(C_h)\cap C'$, there exists  $x'\in C_h$ with  $x\in \delta(x',e)$. Then $x'<y'$, since by hypothesis the partition $\mathcal C$ agrees with the Wheeler order $<$. 
 Finally,   by the Wheeler properties,  $x<y$ follows from  $x'<y'$,  $x\in \delta(x',e)$,  and   $y\in \delta(y',e)$. 
 
 \item[$R(C')$ ]   In this case, let  $ x\in C_1'=C'\setminus \delta_e(C),  y \in C'_2=\delta_e(C_h)\cap C'$. We prove that 
$x<y$.   From  $R(C')$ it follows that  there exists $k<h$ with $C'\subseteq \delta_e(C_k)$, hence, there exists  $x' \in C_k$ with  $x\in \delta(x',e)$. From $ y \in C'_2=\delta(C_h)\cap C'$ it follows that there exists $y'\in C_h$  with  $y\in \delta(y',e)$. From $x'\in C_k$ and $y'\in C_h$ it follows   $x'<y'$,  since   by hypothesis the partition $\mathcal C$ agrees with the Wheeler order $<$.  Finally, $x<y$ follows from $x'<y'$,  $x\in \delta(x',e)$,      $y\in \delta(y',e)$,  and  Wheeler properties.  
\end{itemize}

From the above analysis the thesis follows.
 \end{proof}

\begin{rem}
 The equivalence relation  $\approx_{out}$  corresponding to the output partition $\mathcal C_{out}$ of the Forward Algorithm can be a {\em proper} refinement of the equivalence $\approx_{\mathcal A}$ described in Definition  \ref{sim_A},  as the   automaton in Fig \ref{ex:out_neq_A} shows:  the two last states are 
$\approx_{\mathcal A}$-equivalent and   not $\approx_{out}$-equivalent. 
\end{rem}
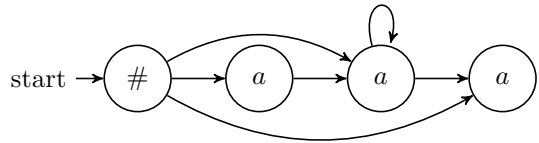
\begin{figure}[!htb]
\centering
 \begin{tikzpicture}[->,>=stealth', semithick, auto, scale=.8]
    \node[state,initial, label=above:{}] (S)    at (-2,0)		{$\#$};
    \node[state, label=above:{}] (Q_12)    at (0,0)		{$a$};
    \node[state, label=above:{}] (Q_3)    at (2,0)				{$a$};
    \node[state, label=above:{}] (Q_4)    at (4,0)				{$a$};
    \draw (S) edge node {} (Q_12);
    \draw (S) edge [bend left=30, above] node {} (Q_3);
    \draw (Q_12) edge node {} (Q_3);
     \draw (Q_3) edge   node {} (Q_4);
    \draw (S) edge [bend right, above] node {} (Q_4);
    \draw (Q_3) edge  [loop above] node {} (Q_3);
\end{tikzpicture} 
\caption{An NFA  for which the output relation $\approx_{out}$ given by the Forward Algortithm  is a proper refinement of    $\approx_{\mathcal A}$.}\label{ex:out_neq_A}
\end{figure}

We are now ready to prove that deciding Wheelerness for reduced NFA is in $P$. 

 \begin{corollary} \label{polyreduced} We can decide in polynomial time whether  a reduced state-labeled NFA $\mathcal A$  admits a Wheeler order. 
 \end{corollary}
   \begin{proof} This follows by the previous lemmas and the uniqueness of the Wheeler order on a reduced NFA (see Lemma \ref{unique}). If we start the Forward Algorithm from a reduced NFA, by Lemma \ref{equireach} we know that 
   the output   partition  $\mathcal C_{out}$ consists  of singleton classes. By Lemma \ref{agree} we also know that if  $\mathcal A$  is Wheeler then the unique possible Wheeler order is given by the (ordered) partition $\mathcal C_{out}$. Hence, to decide whether a reduced NFA $\mathcal A$  is Wheeler we can apply the algorithm, produce $\mathcal C_{out}$  in polynomial time, and test  whether the  induced  order is Wheeler (this can be done in polynomial time, see \cite{alanko2020regular}). 
   \end{proof}
 

Moreover, the Forward Algorithm achieves the following:  if  $\mathcal A / \approx_{out}$ is defined as in Definition \ref{def:A/equiv} (but using relation $\approx_{out}$ instead of $\approx_{\mathcal A}$),  it holds:

\begin{corollary}\label{cor:sort NFA}
Let $\mathcal A$ be a state-labeled NFA. If $\mathcal A$ is Wheeler, then then the Forward Algorithm builds and sorts, in polynomial time, the equivalent Wheeler NFA $\mathcal A / \approx_{out}$.
\end{corollary}
\begin{proof}
By Lemma \ref{equireach}, $\approx_{out}$ is a refinement of $\approx_{\mathcal A}$ (Definition \ref{sim_A}).
Using the same construction of 
Definition \ref{def:A/equiv} and
Lemma \ref{quotient}, we can moreover see that $\mathcal A / \approx_{out}$ (having elements of $\mathcal C_{out}$ as states) is equivalent to $\mathcal A$. By Lemma \ref{agree}, if $\mathcal A$ is Wheeler then $\mathcal C_{out}$ agrees with any Wheeler order $<$ of $\mathcal A$. It easily follows that the order $<_{out}$ defined by $C_i <_{out} C_j$ if and only if $i<j$ is a Wheeler order on $\mathcal A / \approx_{out}$. To see this, first note that if $\lambda(C_i) \prec \lambda(C_j)$ then $C_i <_{out} C_j$ since the Forward Algorithm preserves the order of the labels (Wheeler (i)). To prove Wheeler (ii), let $C_i <_{out} C_j$ and $C_{i'}, C_{j'}$ be successors of $C_{i}$ and $C_{j}$, respectively, such that $\lambda(C_{i'}) = \lambda(C_{j'})$. Then, by definition of $\mathcal A / \approx_{out}$ there exist $u\in C_i$, $v\in C_j$, $u'\in C_{i'}$, and $v'\in C_{j'}$ such that $u',v'$ are successors of $u,v$, respectively, with  $\lambda(u') = \lambda(v') = \lambda(C_{i'}) = \lambda(C_{j'})$. Since $C_i <_{out} C_j$, by Lemma \ref{agree} we have that $u<v$. By Wheeler (ii) on $\mathcal A$, it follows that $u' < v'$. Then, it must be the case that $C_{i'} <_{out} C_{j'}$: if this were not the case, i.e. if $C_{j'} <_{out} C_{i'}$, then by Lemma \ref{agree} we would have $v'<u'$, a contradiction. It follows that also Wheeler (ii) holds, therefore $\mathcal A / \approx_{out}$ is Wheeler with order $<_{out}$.
\end{proof}

Corollary \ref{cor:sort NFA} allows us to circumvent the NP-completeness of the problem of recognizing and sorting general Wheeler NFA \cite{DBLP:conf/esa/GibneyT19}. This does not mean that we break the problem's NP-completeness: while a Wheeler $\mathcal A$ induces a Wheeler $\mathcal A / \approx_{out}$ by the Forward Algorithm, the opposite is not true. As shown in Figure \ref{ex:nonW->W}, there exist non-Wheeler NFA $\mathcal A$ such that $\mathcal A / \approx_{out}$ is Wheeler.

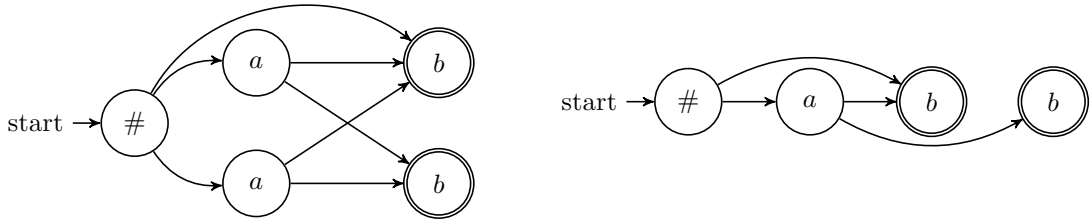
\begin{figure}[!htb]
\centering
\begin{minipage}{.45\textwidth}
    \begin{tikzpicture}[->,>=stealth', semithick, auto, scale=.8]
    \node[state,initial, label=above:{}] (S)    at (-2,0)		{$\#$};
    \node[state, label=above:{}] (Q_1)    at (0,1)		{$a$};
    \node[state, label=above:{}] (Q_2)    at (0,-1)		{$a$};
    \node[state, label=above:{}, accepting] (Q_3)    at (3,1)				{$b$};
    \node[state, label=above:{}, accepting] (Q_4)    at (3,-1)				{$b$};
    \draw (Q_1) edge node {} (Q_3);
    \draw (Q_2) edge node {} (Q_4);
    \draw (Q_1) edge node {} (Q_4);
    \draw (Q_2) edge node {} (Q_3);
    \draw (S) edge [bend left=50, above] node {} (Q_3);
    \draw (S) edge [bend left, above] node {} (Q_1);
    \draw (S) edge [bend right, above] node {} (Q_2);
\end{tikzpicture} 
\end{minipage}
\begin{minipage}{.45\textwidth}
    \begin{tikzpicture}[->,>=stealth', semithick, auto, scale=.8]
    \node[state,initial, label=above:{}] (S)    at (-2,0)		{$\#$};
    \node[state, label=above:{}] (Q_12)    at (0,0)		{$a$};
    \node[state, label=above:{}, accepting] (Q_3)    at (2,0)				{$b$};
    \node[state, label=above:{}, accepting] (Q_4)    at (4,0)				{$b$};
    \draw (S) edge node {} (Q_12);
    \draw (S) edge [bend left=30, above] node {} (Q_3);
    \draw (Q_12) edge node {} (Q_3);
    \draw (Q_12) edge [bend right, above] node {} (Q_4);
\end{tikzpicture} 
\end{minipage}
\caption{Left: non-Wheeler NFA $\mathcal A$ (the two states labeled $b$ cannot be ordered). Right: Wheeler NFA $\mathcal A / \approx_{out}$ and corresponding order output by the Forward Algorithm. The two states labeled $a$ have been merged into a single state.}\label{ex:nonW->W}
\end{figure}

\section{Closure Properties for Wheeler Languages}

In this section we classify operations on languages depending on  whether they preserve Wheelerness or not. The first observation is that Wheeler languages, being a subclass of the class of Ordered Languages (see \cite{thierrin1974ordered_aut}), are star-free  (that is, they can be generated from finite languages by Boolean operations and compositions only). As such,  they can be  definable in the so-called the first order theory of linear orders $FO(<)$. However, as we shall see, there are very few ``classical'' operations which preserve Wheeler Languages.  

%
 
   \subsection{Booleans}  \label{booleans}

\begin{lemma} \label{bool_pos} ~
 \begin{enumerate}
 \item  Finite and co-finite languages are Wheeler. 
  \item The union of a Wheeler language with a finite set is Wheeler. 
 \item  The intersection of two Wheeler languages is Wheeler. 
 \item If $\mathcal L$ is Wheeler, then $\pf L$ is Wheeler. 
 \item If $\mathcal L$ is Wheeler,  then $\pf L \setminus \mathcal  L$ is Wheeler.
 \end{enumerate}
\end{lemma}

 \begin{proof}
 \noindent
  \begin{enumerate}
 \item
  If $\mathcal L$ is finite (co-finite) and  $(\alpha_i)_{i\geq 1}$ is a monotone  sequence    in $\pf L$, then there is a $k\geq 1$ such that, for all $i>k$,  the length of $\alpha_i$ is longer than the length of any word in $\mathcal L$ ($\overline {\mathcal L}$, respectively). 
  This shows that any word having $\alpha_i$   as prefix  does not belong to $\mathcal L$ ($\overline {\mathcal L}$, respectively), so that 
 $\alpha_i   \equiv_{\mathcal L} \alpha_j$ for all $i,j>k$, and $\mathcal L$ is Wheeler    by  Lemma
 \ref{infinite_nonwh}.
 \item If $\mathcal L$ is Wheeler, $F$ is a finite set,  and $(\alpha_i)_{i\geq 1}$ is a   a monotone  sequence    in $\pf {L\cup F}$,  then there is a $k\geq 1$ such that, for all $i>k$, 
  the length of $\alpha_i$ is longer than the length of any word in the finite set $F$.  This implies that  $\alpha_i\in \pf L$, for $i>k$,  and, since $\mathcal L$ is Wheeler,    there exists $h\geq k$ such that  $\alpha_j\equiv_{\pf L}\alpha_{j+1}$,  for all $j\geq h$.  Then $\alpha_j\equiv_{\pf {L\cup F}}\alpha_{j+1}$, because, for length reasons, no words in $F$ can have an $\alpha_j$ as prefix. 

  \item  Suppose $\mathcal L_1,  \mathcal L_2$ are Wheeler, and consider  a monotone sequence $(\alpha_i)_{i\geq 1}$  in $\pf {{\mathcal L}_1\cap {\mathcal L}_2}$.   Since \[\pf {{\mathcal L}_1\cap {\mathcal L}_2}\subseteq \pf {{\mathcal L}_1}\cap \pf  {{\mathcal L}_1}  \]  and $\mathcal L_1,  \mathcal L_2$ are Wheeler,  by  Lemma \ref{infinite_nonwh}  there exists $h$ such that
 $ \alpha_i\equiv_{\mathcal L_1}\alpha_j$, and  $\alpha_i\equiv_{\mathcal L_2}\alpha_j $
  both hold for $i,j>h$. 
It follows that 
  $\alpha_i  \equiv_{\mathcal L_1\cap \mathcal L_2} \alpha_{j}$, for all $i,j>h$, and so $\mathcal L_1\cap \mathcal L_2$ is Wheeler by Lemma \ref{infinite_nonwh}.
  
\item  Obvious, by considering a $WDFA$ recognizing $\mathcal L$, and  considering all states as final.
    \item  Obvious, by considering a $WDFA$ recognizing $\mathcal  L$, and   changing non final with final states.  

  \end{enumerate}
 \end{proof}

\begin{corollary}\label{oneletter} The only  Wheeler Languages on the one letter alphabet $\Sigma=\{a\}$
are the finite or co-finite ones. 
\end{corollary}
\begin{proof}  
Suppose  $\mathcal L\subseteq \{a\}^*$ is  neither finite nor co-finite. Since  $\mathcal L $ is not finite and the alphabet contains only one letter, we have $\pf L=\Sigma^*$, and, 
since $\mathcal L$ is not co-finite, we have that  $\pf L\setminus {\mathcal L}=\Sigma^*\setminus {\mathcal L}$ is infinite. 
Let $\alpha=\alpha_1$ be a word in $\mathcal L$. Since there are only a finite number of words which are co-lexicographically smaller than $\alpha$, there exists $\alpha_2\in \pf L\setminus \mathcal L$ such that 
$\alpha_1\prec \alpha_2$.  Suppose we already have 
\[\alpha_1\prec\alpha_2\prec \ldots \prec \alpha_m,\]
 $m$ even, with $\alpha_i\in \mathcal  L$,  for odd $i$'s, and $\alpha_i \not \in \mathcal L$ for even $i$'s. 
Then, since $\mathcal L$ is infinite and there are only a finite number of words which are co-lexicographically smaller than $\alpha_m$, there exists $\alpha_{m+1}\in \mathcal L$ such that
$\alpha_m\prec \alpha_{m+1}$. 
Hence, 
we can define a monotone sequence which is not eventually constant modulo $\equiv_{\mathcal L}$, and $\mathcal L$ is not Wheeler by Lemma \ref{infinite_nonwh}.
 \end{proof}
 
 We now turn to boolean operation \emph{not}  preserving Whelerness:
 
  \begin{lemma} \label{neg}
 Wheeler Languages are not closed for:
 \begin{itemize}
 \item[-] Unions.
 \item [-] Complements.
 \end{itemize}
 \end{lemma}
 \begin{proof}
 {\bf Unions}. ~~The languages $\mathcal L_1=ax^*b, ~ \mathcal L_2=cx^*d$ are easily seen to be Wheeler, but their union is not  (see  Example \ref{nw}).
 
 \medskip
 
 {\bf Complements}. ~~ Let $\Sigma=\{a,b\}$ and $\mathcal L= b^*$. Then $\mathcal L$ is easily seen to be Wheeler, but its complement 
 \[\overline {\mathcal L}=  \{\alpha\in \Sigma^*~:~ \alpha~ \text{contains at least an occurrence of the letter}~ a\}\]
 is not Wheeler: consider the monotone sequence ${\text{\em Pref}(\overline {\mathcal L})}$ given by 
 \[
\alpha_i=
\begin{cases}
b^i  &\text {if $i$ is odd;}\\
ab^i &\text{if $i$ is even.}
\end{cases}
\]

If $i$ is odd,    $\alpha_i=b^i\not \in \overline {\mathcal L}$, while $\alpha_{i+1}=ab^{i+1}\in\overline{ \mathcal L}$, so that $\alpha_i  \not  \equiv_{\overline{\mathcal L}} \alpha_{i+1}$, and $\overline {\mathcal L}$ is not Wheeler by Lemma \ref{infinite_nonwh}.

\medskip

 \end{proof}

 \subsection{Concatenation}
 
In general, the concatenation of two Wheeler languages is not necessarily Wheeler, as the following example shows:

\begin{example}The languages $\mathcal L_1=b^*a,~\mathcal L_2=b^+a$ are easily seen to be Wheeler, but their concatenation $\mathcal L= \mathcal L_1 \cdot \mathcal L_2 $ is not:  
  consider the monotone sequence in $(\pf L, \preceq)$ given by
  
   \[
\alpha_i=
\begin{cases}
ab^ia &\text {if $i$ is odd;}\\
b^{i}a &\text{if $i$ is even.}
\end{cases}
\]

 If $i$ is odd, we have  $\alpha_i\in \mathcal L$,  while $\alpha_{i+1}\not \in  \mathcal L$. Hence, $\alpha_i\not \equiv_{\mathcal L} \alpha_{i+1}$ for infinite $i$'s, and $\mathcal L$ is not Wheeler. 
\end{example}

On the positive side,  we   prove that the \emph{right} concatenation of a Wheeler language with a finite set, is Wheeler.  This is not true if consider \emph{left} concatenation, even if the finite set is a single-letter word, as the following example shows.

\begin{example}The language  $\mathcal L =\{a^i:i \geq 1\}\cup \{ba^ib: i\geq 1\}$  is  easily seen to be Wheeler but  its concatenation on the left with the letter $c$  is not. 
Indeed  $c\cdot \mathcal L =\{ca^i:i \geq 1\}\cup \{cba^ib\}$ and there exists a monotone sequence in ${\text{\emph{Pref}}(c\cdot \mathcal L)}$ which is not eventually constant modulo $c\cdot \mathcal L$:
  \[ca\succ cba\succ  caa\succ\ldots\succ cba^i \succ ca^{i+1}\succ\ldots\]
From Lemma  \ref{infinite_nonwh},  it follows that $c\cdot \mathcal L$ is not Wheeler. 
\end{example}

\begin{lemma} If $\mathcal L$ is Wheeler and $F$ is a finite set, then  ${\mathcal L}\cdot F$  is  Wheeler.
\end{lemma}
  \begin{proof}
  Suppose $\mathcal L$ is Wheeler, $F$ is a finite set,  and $n=max\{|w|: w\in F\}$ is the maximum of all lengths of words in $F$.   If $(\alpha_i)_{i\geq 1}$ is a  monotone  sequence    in $\pf {L\cdot F}$,  then by possibly  erasing an initial finite sequence we may suppose w.l.o.g. that $|\alpha_i|\geq 2n$, 
 and  all $\alpha_i$ end with the same $2n$-suffix  $\gamma_1\gamma_2$, with $|\gamma_1|=|\gamma_2|=n$. Let $\alpha_i', \alpha_i''\in \pf L$ be such that 
  \[\alpha_i=\alpha_i''\gamma_1\gamma_2=\alpha_i'\gamma_2.\]
 Then  both $(\alpha_i')_{i\geq 1}$ and  $(\alpha_i'')_{i\geq 1}$  are   monotone  sequences    in $\pf {L}$ and, since $\mathcal L$ is Wheeler,  there exists $k$ such  that $\alpha_i'\equiv_{\mathcal L} \alpha_j'$ and $\alpha_i''\equiv_{\mathcal L} \alpha_j''$, for all $i,j\geq k$. We  next  prove that, for all $i,j\geq k$, we also have $\alpha_i'\equiv_{\mathcal L\cdot F} \alpha_j'$, from which $\alpha_i\equiv_{\mathcal L\cdot F} \alpha_j$ follows. We must prove that for all $ \beta $, $ \alpha'_{i}\beta \in \mathcal L\cdot F \Leftrightarrow \alpha'_{j}\beta \in \mathcal L\cdot F $.  
 Suppose   $\alpha_i'\beta \in \mathcal L\cdot F$. Then \[\alpha_i'\beta = \alpha_i'' \cdot \beta' \cdot f,\] with  $\alpha_i'' \cdot \beta' \in \mathcal L$ and $f\in F$. 
From $\alpha_i''\equiv_{\mathcal L} \alpha_j''$ it follows  $\alpha_j'' \cdot \beta' \in \mathcal L$,  so that 
  \[\alpha_j'\beta = \alpha_j'' \cdot \beta' \cdot f\in \mathcal L\cdot F\] 
  Summarizing, we proved that all elements of  the monotone sequence  $(\alpha_i)_{i\geq 1}$ end eventually in the same $\equiv_{\mathcal L\cdot F}$-class, hence $\mathcal L\cdot F$ is Wheeler. 

  \end{proof}
 
 \subsection{Kleene Star} 
In general, Wheeler languages are not closed for Kleene star, as the following example shows.

\begin{example}
 The language $\mathcal L=\{aa\}$ is Wheeler (as any finite language), but  $\mathcal L^*= \{a^{2i+2}: i\geq 0\}$ is not Wheeler (see Example \ref{nw}). 
\end{example}

On the other hand,  we can characterise    which words $\alpha$ have a Kleene star $\alpha^*$ which is Wheeler, and, more generally, when a regular language  of the form $\alpha_1\alpha^*\alpha_2$ is Wheeler.

\begin{definition}
We say that $ \alpha \in \Sigma^* $ is \emph{primitive} if there exists no $ \beta \neq \epsilon $ and $ i>1 $, such that $ \alpha = \beta^{i} $.  
\end{definition}

Primitive words are important for Wheeler automata and languages as  seen in the following results. 

\medskip

 \begin{lemma} If $\mathcal A=(Q,s,\delta, F,<)$ is a WDFA  and  $\alpha $ is the label of a simple cycle in $\mathcal A$, then $\alpha$ is primitive.
 \end{lemma}
 \begin{proof}
  Suppose, by way of a contradiction, that there exists a  simple cycle labelled by $\alpha$ and there exists  $ i>1 $, such that $ \alpha = \beta^{i} $.  Then there exists $n<m<r$ such that   $\beta^n, \beta^r$ are both  labels of   cycles starting from the same  vertex $u$, while $\delta(u,\beta^m)\neq u$.  Let $\gamma$ be a word such that $\delta(s,\gamma)=u$. Consider the sequence 
  $(\gamma\beta^h)_{h\in \mathbb N}$, and  note that it  is monotone: if $\gamma\prec \gamma \beta$, then $\gamma\beta^k \prec \gamma\beta^{k+1}$ holds for any $k$, and similarly by transitivity of $\prec$ we obtain that $\gamma\beta^k \prec \gamma\beta^h$ holds for any $h>k$. Thus, the sequence is monotonically increasing. Conversely, if $\gamma \succ \gamma\beta$ then the sequence is monotonically decreasing.
  It follows that either $\gamma\beta^n\prec \gamma\beta^m\prec  \gamma\beta^r$, or $\gamma\beta^n\succ \gamma\beta^m \succ \gamma\beta^r$.
  Since    
$\gamma\beta^n, \gamma\beta^r \in I_u$,  by Lemma  \ref{convex_sets} we should also have $\gamma\beta^m\in I_q$. 
\end{proof}

We shall use the following:

 \begin{notation}
 $ \alpha'  \vdash \alpha $ stands for $ \alpha' $ is a prefix of $ \alpha $   and
$ \alpha'  \dashv \alpha $ stands for $ \alpha' $ is a  suffix of $ \alpha $.
 \end{notation}

\begin{lemma} Let $\alpha_1, \alpha, \alpha_2\in \Sigma^*$. Then 
\[ \alpha_1  \alpha^* \alpha_2~ \hbox{ is Wheeler}~ \Leftrightarrow \alpha~ \hbox{is primitive}.\]
\end{lemma}
\begin{proof}
($\Rightarrow$) Suppose $\alpha$ is not primitive, say $\alpha =\beta^k$ with $k>1$, $\beta\neq \epsilon$, and consider the sequence
\[\alpha_1\beta^k , \alpha_1\beta^{k+1},~  \alpha_1\beta^{2k} , \alpha_1\beta^{2k+1},~\alpha_1\beta^{3k}, \alpha_1\beta^{3k+1}, \ldots\]
in $\pf {\alpha_1  \alpha^* \alpha_2}$. As in the previous lemma, we can prove that the sequence is monotone: if $\alpha_1\prec \alpha_1\beta$,   the sequence is monotonically increasing, while, if $\alpha_1 \succ \alpha_1\beta$ then the sequence is monotonically decreasing.
However, this sequence does  not become eventually constant modulo $\equiv_{\alpha_1  \alpha^* \alpha_2}$, because, for all $n$, 
 \[ \alpha_1\beta^{nk}\alpha_2\in \alpha_1  \alpha^* \alpha_2 ~\hbox{ while }~~ \alpha_1\beta^{nk+1}\alpha_2\not  \in \alpha_1  \alpha^* \alpha_2\]
 From the above and Lemma \ref{infinite_nonwh} it follows that $ \alpha_1  \alpha^* \alpha_2$ is not Wheeler. 
 
 \medskip
 
 ($\Leftarrow$)  If $\alpha$ is primitive we first show  that   $ \alpha_1  \alpha^*  $ is   Wheeler. Suppose not. Then there is a monotone sequence in $\pf {\alpha_1  \alpha^* }$ which does not become eventually constant modulo $\equiv_{\alpha_1\alpha^*}$. By erasing an opportune prefix  of the sequence we may suppose that it has  the form
 \[\alpha_1\alpha^{h_1}\beta_1,~ \alpha_1\alpha^{h_2}\beta_2,~ \alpha_1\alpha^{h_3}\beta_3, \ldots\]
 with $\beta_i\vdash \alpha$, for all $i$, and  that  all elements of the sequence end with the same $3|\alpha|$ characters.
  Notice that, since the sequence is not eventually constant modulo $\equiv_{\alpha_1\alpha^*}$, there  must be  infinite $i$'s such that $\beta_i\neq \beta_{i+1}$. 
 Hence,   there are two different $\alpha$-prefixes, $\beta, \beta'$ such that  $\alpha^3 \beta$  and $\alpha^3 \beta'$  end with the same $3|\alpha|$-characters, which implies that there exists an $\alpha$-prefix $\gamma$ such that $\alpha$ and $\alpha\gamma$ end with the same $|\alpha|$-characters; but then  there exists $\delta$ such that 
 $\alpha=\delta \gamma$, where $\delta, \gamma$  are both proper prefixes and proper suffixes of $\alpha$. This implies $\alpha=\delta\gamma=\gamma\delta$ which in turn implies (see (\cite{lyndon1962})) that $\alpha$ is not primitive, a contradiction. 
 
 Hence,   If $\alpha$ is primitive  then  $ \alpha_1  \alpha^*  $ is   Wheeler, and $\alpha_1\alpha^*\alpha_2$ is also Wheeler, being a concatenation of a Wheeler language with a finite set on the right. 
\end{proof}

 \subsection{Factors, Suffixes, and Inverses}
 Wheeler Languages are not closed for factors, suffixes, or inverses:
  \begin{example} \label{neg2}
{\bf Factors and Suffixes}. The language $\mathcal L_1=ax^*b~|~zx^*d$ is   Wheeler (see Example \ref{exw}), but $\mathcal L=\fc {L_1}$ is not:
 consider the monotone sequence in $(\pf L, \preceq)$ given by
  
   \[
\alpha_i=
\begin{cases}
x^i &\text {if $i$ is odd;}\\
 ax^i&\text{if $i$ is even.}
\end{cases}
\]

 if $i$ is odd,  $\alpha_i\not \equiv_{\mathcal L} \alpha_{i+1}$, because $\alpha_id=x^id\in \mathcal L$ whereas $\alpha_{i+1}d=ax^{i+1}d \not \in \mathcal L$; 
hence 
$\mathcal L=\fc {L_1}$ is not Wheeler by Lemma \ref{infinite_nonwh}. 
Similarly, $\sff L$ is not Wheeler: considering the same  monotone  sequence    above  we have $\alpha_i \in {\text{\em Pref}( {\sff L})}$ and $\alpha_i\not \equiv_{\sff L} \alpha_{i+1}$,  for odd $i$'s, because $\alpha_id=x^id\in \sff L$ whereas $\alpha_{i+1}d=ax^{i+1}d \not \in \sff L$.

\medskip

{\bf Inverses}. Suppose, by way of a contradiction, that Wheeler languages were closed under inverses, and consider again the Wheeler language $\mathcal L=ax^*b~|~zx^*d$; then, by Lemma \ref{bool_pos},   ${\text{\em Pref}(  {  {\mathcal L}^{-1}})}^{-1}=\sff L$  would be Wheeler, while we proved the  opposite  in the previous point.
\end{example}

\subsection{Morphisms}
 
 We now consider preservation under inverse image of monoid morphisms.     Wheeler Languages are not closed in general under   inverse images of   morphisms. E.g.   consider \[\Sigma=\{a,b,c,d,x\},~~ \Sigma'=\{a,b,d,x,z\}, ~~ \mathcal L=ax^*b~|~zx^*d\subseteq \Sigma'^*,\] and the morphism  $\phi$  defined by 
  $\phi(c)=z$  and the identity on the other letters. Then $ \mathcal L$ is   Wheeler (see Example \ref{exw}), while $ \phi^{-1}(\mathcal L)=ax^*b~|~cx^*d$ is not Wheeler (see Example  \ref{nw}). We next prove that Wheeler languages are  closed under inverse images of co-lex monoid morphisms:

 \begin{definition}   Let $\Sigma, \Sigma'$ be two finite alphabet.
 A  {\em co-lex morphism}  between $(\Sigma^*,\cdot,\preceq)$,   $(\Sigma'^*, \cdot,\preceq)$ is a monoid morphism $\phi:\Sigma^*\rightarrow \Sigma'^*$  such that 
 \[\alpha\preceq \alpha' \Rightarrow \phi(\alpha) \preceq \phi(\alpha')\]
 \end{definition}

\begin{lemma}Suppose $\Sigma, \Sigma'$ are finite alphabets and $\phi: (\Sigma^*,\preceq) \rightarrow (\Sigma'^*,\preceq)$ is a co-lex morphism.  If $\mathcal L\subseteq \Sigma'^*$ is a Wheeler language, then  $\phi^{-1}(\mathcal L)\subseteq \Sigma^*$ is a Wheeler language.
\end{lemma}
\begin{proof}
  If   $\phi:(\Sigma^*,\preceq) \rightarrow (\Sigma'^*,\preceq)$ is a morphism of ordered monoids and   $\phi^{-1}(\mathcal L)$ is not Wheeler, we prove that  $\mathcal L$ is not Wheeler. Since regular languages are closed by inverse images of morphisms,  $\phi^{-1}(\mathcal L)$ is a regular,  non Wheeler language;  by Lemma \ref{infinite_nonwh} there exists a strictly  monotone sequence $(\gamma_i)_{i\in \mathbb N}$ in $\pf {\phi^{-1}(\mathcal L)}$ with $\gamma_i\not \equiv_{\phi^{-1}(\mathcal L)} \gamma_{i+1}$. Since $\phi$ is a morphism,   we obtain $\phi(\gamma_i)\in \pf L$. Moreover, since $\phi$ is a co-lex morphism, we obtain that  $(\phi(\gamma_i))_{i\in \mathbb N}$ is monotone  and  $\phi(\gamma_i)\not \equiv_{\mathcal L} \phi(\gamma_{i+1})$ for all $i$.    Hence, $(\phi(\gamma_i))_{i\in \mathbb N}$ is strictly monotone and  Lemma \ref{infinite_nonwh} implies that $\mathcal L$ is not Wheeler.
\end{proof}

 The closure of Wheeler languages under the inverse image of co-lex morphisms may  suggest a natural generalization of the algebraic characterization of regular languages.   Remember that  a language $\mathcal L\subseteq \Sigma^*$ is said to be {\em recognized by a monoid morphism} $\phi:(\Sigma^*, \cdot)\rightarrow (M, \cdot)$ if $\mathcal L=\phi^{-1}(\phi(\mathcal L))$ (or, equivalently, if $\alpha \in \mathcal L$ and $\phi(\alpha)=\phi(\beta)$ implies $\beta\in \mathcal L$).   The algebraic characterization of regular languages states that these languages are exactly    the ones which are recognized by  morphisms over finite monoids.
 
 Suppose now    we add a total order $\leq$ over the elements of the monoid $M$; we say that a  monoid morphism $\phi:(\Sigma^*, \cdot)\rightarrow (M, \cdot)$ {\em respect } the corresponding orders $\prec, \leq$ 
if,  for all $\alpha, \beta\in \Sigma^*$ it holds:
\[ \alpha\preceq \beta \Rightarrow \phi(\alpha)\leq \phi(\beta).\]

\begin{lemma}\label{pseudochar}
If a language $\mathcal L\subseteq \Sigma^*$ is  recognized by   a morphism  over a  finite monoid $(M,\cdot)$ and $\leq$ is an order over $M$ such that $\phi$ respect  the orders $\prec, \leq$, then $\mathcal L$ is Wheeler. 
\end{lemma}
 \begin{proof}  $\mathcal L$   is regular, since it is  recognized by   a morphism  over a  finite monoid $(M,\cdot)$. Suppose it is not Wheeler. Then by Lemma  \ref{infinite_nonwh} there exists a monotone (say increasing) sequence $(\alpha_i)_{i\in \mathbb N}$  in $\pf L$  which is not eventually constant.
Since the morphism respect the order, we have $\phi(\alpha_i)\preceq \phi(\alpha_{i+1})$, for all $i$. Moreover, $\phi(\alpha_i)\neq\phi(\alpha_{i+1})$, for every $i$ such that  $\alpha_i\not \equiv_{\mathcal L} \alpha_{i+1}$:   from the previous inequality it follows that  there exists $\delta\in \Sigma^*$ with $\alpha_i \delta\in \mathcal L$ and $\alpha_{i+1}  \delta\not \in \mathcal L$  (or viceversa); if $\phi(\alpha_i)=\phi(\alpha_{i+1})$ then 
\[\phi(\alpha_i\delta)=\phi(\alpha_i)\phi(\delta)=\phi(\alpha_{i+1})\phi(\delta)=\phi(\alpha_{i+1}\delta),\]
and from $\alpha_i \delta\in \mathcal L$   it then follows $\alpha_{i+1}\delta\in  \mathcal L$, a contradiction. 
Hence, $( \phi(\alpha_i)_{i\in \mathbb N}$ should be a monotone sequence which is strictly increasing for infinitely many index $i$, which is impossible, since $M$ is finite. 

\end{proof}

Unfortunately, Lemma \ref{pseudochar} is too strong and cannot be reversed:  the class of  languages which are recognized by morphism as in Lemma  \ref{pseudochar} is closed under complements and factors, while Wheeler languages are not.

\subsection{Intervals} 
 \begin{definition}  If  $\alpha_0\preceq \alpha_1 \in \Sigma^+$,   we define    the  {\em intervals }  $(\alpha_0, \alpha_1), [\alpha_0, \alpha_1), (-\infty , \alpha_1)$ \ldots based on $\alpha_0, \alpha_1$  as usual, e.g.:
 \[ (\alpha_0,\alpha_1 )=\{\beta\in \Sigma^*: \alpha_0\prec \beta \prec \alpha_1\},~~ [\alpha_0,\alpha_1 )=\{\beta\in \Sigma^*: \alpha_0\preceq \beta \prec \alpha_1\}, ~~  (-\infty,\alpha_1 )=\{\beta\in \Sigma^*: \beta \prec \alpha_1\}, \ldots \] 
\end{definition}

\begin{lemma}\label{intervals} Suppose $\alpha_0\preceq \alpha_1 \in \Sigma^+$. Then  all intervals based on  $\alpha_0, \alpha_1$  are Wheeler.
\end{lemma}

 \begin{proof}
Let $\alpha_1\in \Sigma^+$, and  consider the interval   $(-\infty,\alpha_1 )$.  If  $F=\{\beta : \beta \prec \alpha_1, |\beta|= |\alpha_1|\}$ we have 
\[(-\infty,\alpha_1 )=\Sigma^+\cdot F \cup \{\gamma: \gamma\prec \alpha_1, |\gamma|\leq |\alpha_1|\}\]
which is   Wheeler  by Lemma \ref{bool_pos}.

Similarly, 
\[(\alpha_0, +\infty )=\Sigma^*\cdot  \{\beta: \alpha_0\prec \beta,  |\beta|\leq |\alpha_0|\}, \] 
is Wheeler.
 If  $\alpha_0\prec \alpha_1$,  then the interval $(\alpha_0, \alpha_1)=(\alpha_0, +\infty)\cap (-\infty,\alpha_1)$ is   Wheeler, as intersection of Wheeler languages. 
Finally, the (half-)closed intervals $(-\infty,\alpha_0]$, $[\alpha_0, \alpha_1),  (\alpha_0, \alpha_1], \ldots$ are obtained from the  open versions by adding  one or two words, hence they are Wheeler by  Lemma \ref{bool_pos}.
\end{proof}

Note that Wheelerness does not generalize from interval to convex sets, as the following example shows.

\begin{example} 
The regular language \[ \mathcal L~=~ ax^*a~|~bx^*b~|~ b\]
is convex in $\pf L$ but it is not Wheeler. 
\end{example}

\section{Conclusions and Open Problems}

 Wheeler Languages represent a formal tool to elegantly and fruitfully cast the notion of ordering of strings of a regular language $ \mathcal L $ on an ordering of the states of an automaton $ \mathcal A $ recognising $ \mathcal L $. The key property, made explicit by the definition of Wheeler graphs, allows to doubly-link the co-lexicographic order of strings read while reaching a state $ q $ with the position of $ q $ in the Wheeler order of $ \mathcal A $'s states. This is obtained by the initial fixing of an ordering of the alphabet $ \Sigma $, which is the marking difference between the approach on ordering of states developed here and the work on ordered automata carried out in \cite{thierrin1974ordered_aut}.
 
 \medskip        
 
Many questions remain open, especially on the operational characterisation of Wheeler languages. Among the  problems left open we mention:
 
\begin{enumerate}
\item   Theorem \ref{decidability} allow us to prove that the problem of deciding a regular language  accepted by a given finite deterministic automaton   is Wheeler in polynomial time.  Can we generalise this theorem to NFA's, in order to 
show that we can decide in polynomial time if a regular languages accepted by a $NFA$  is Wheeler?
\item Is there a natural fragment of $FO(<)$ describing Wheeler Languages, or, is there a natural \emph{logic} describing Wheeler Languages?
\item Can we find a finite number of ``Wheeler operations'' and a finite number of  ``basic Wheeler Languages'' such that all Wheeler languages are obtained from the basic ones using the Wheeler operations?
\item Can we characterise Wheeler languages using monoids or other algebraic structures?
\end{enumerate}   

\

{\bf Acknowledgements}. We thank Davide Martincigh for careful reading of this paper and for elegant suggestions.

\bibliographystyle{alpha}
\bibliography{Arxiv_Wheeler_Languages}

\begin{thebibliography}{EGMT19}

\bibitem[ADPP20]{alanko2020regular}
Jarno Alanko, Giovanna D'Agostino, Alberto Policriti, and Nicola Prezza.
\newblock Regular languages meet prefix sorting.
\newblock In {\em Proceedings of the 2020 ACM-SIAM Symposium on Discrete
  Algorithms}, pages 911--930, 2020.

\bibitem[BW94]{Burrows94ablock-sorting}
Michael Burrows and David Wheeler.
\newblock A block-sorting lossless data compression algorithm.
\newblock Technical report, DIGITAL SRC RESEARCH REPORT, 1994.

\bibitem[DG08]{DBLP:conf/birthday/DiekertG08}
Volker Diekert and Paul Gastin.
\newblock First-order definable languages.
\newblock In J{\"{o}}rg Flum, Erich Gr{\"{a}}del, and Thomas Wilke, editors,
  {\em Logic and Automata: History and Perspectives [in Honor of Wolfgang
  Thomas]}, volume~2 of {\em Texts in Logic and Games}, pages 261--306.
  Amsterdam University Press, 2008.

\bibitem[EGMT19]{equi2019complexity}
Massimo Equi, Roberto Grossi, Veli M{\"a}kinen, and Alexandru~I Tomescu.
\newblock On the complexity of string matching for graphs.
\newblock In {\em 46th International Colloquium on Automata, Languages, and
  Programming (ICALP 2019)}. Schloss Dagstuhl-Leibniz-Zentrum fuer Informatik,
  2019.

\bibitem[EMT20]{equi2020conditional}
Massimo Equi, Veli M{\"a}kinen, and Alexandru~I Tomescu.
\newblock Conditional indexing lower bounds through self-reducibility.
\newblock {\em arXiv preprint arXiv:2002.00629}, 2020.

\bibitem[GMS17]{gagie2017wheeler}
Travis Gagie, Giovanni Manzini, and Jouni Sir{\'e}n.
\newblock {Wheeler graphs: A framework for BWT-based data structures}.
\newblock {\em Theoretical computer science}, 698:67--78, 2017.

\bibitem[GT19]{DBLP:conf/esa/GibneyT19}
Daniel Gibney and Sharma~V. Thankachan.
\newblock On the hardness and inapproximability of recognizing wheeler graphs.
\newblock In Michael~A. Bender, Ola Svensson, and Grzegorz Herman, editors,
  {\em 27th Annual European Symposium on Algorithms, {ESA} 2019, September
  9-11, 2019, Munich/Garching, Germany.}, volume 144 of {\em LIPIcs}, pages
  51:1--51:16. Schloss Dagstuhl - Leibniz-Zentrum f{\"{u}}r Informatik, 2019.

\bibitem[LS62]{lyndon1962}
R.~C. Lyndon and M.~P. Schützenberger.
\newblock The equation $a^m=b^nc^p$ in a free group.
\newblock {\em Michigan Math. J.}, 9(4):289--298, 12 1962.

\bibitem[NM07]{DBLP:journals/csur/NavarroM07}
Gonzalo Navarro and Veli M{\"{a}}kinen.
\newblock Compressed full-text indexes.
\newblock {\em {ACM} Comput. Surv.}, 39(1):2, 2007.

\bibitem[ST74]{thierrin1974ordered_aut}
H.-J. Shyr and G.~Thierrin.
\newblock Ordered automata and associated languages.
\newblock {\em Tamkang J. Math}, 5:9--20, 1974.

\end{thebibliography}
\end{document}